\documentclass[journal,twocolumn,final]{IEEEtran}

\usepackage{comment}
\usepackage{amsfonts,amssymb,amsmath,mathtools}
\usepackage{dsfont}
\usepackage{graphicx}
\usepackage{bm}
\usepackage{cite}
\usepackage{bigstrut}
\usepackage{float}
\usepackage{balance}
\usepackage{lipsum}
\usepackage{psfrag}
\usepackage{xcolor}
\usepackage{xfrac}
\usepackage{mathrsfs}
\usepackage{tabularx,array}
\usepackage{microtype}

\usepackage{algorithm}
\usepackage{algpseudocode}
\usepackage[caption=false,font=footnotesize]{subfig}

\usepackage{tikz}
\usetikzlibrary{
    arrows,
    decorations.pathreplacing,
    fadings,
    fit
}

\usepackage{pgfplots}
\usepackage{pgfplotstable}
\pgfplotsset{compat=1.18}
\usepgfplotslibrary{patchplots}

\usepackage{url}
\usepackage{hyperref}
\hypersetup{hidelinks}
\newtheorem{theorem}{Theorem}
\newtheorem{lemma}{Lemma}
\newenvironment{proof}
{\begin{IEEEproof}}
{\end{IEEEproof}}
\newtheorem{proposition}{Proposition}
\newtheorem{remark}{Remark}

\newcommand{\CN}{\mathcal{CN}}
\newcommand{\norm}[1]{\left\lVert#1\right\rVert}

\newcommand{\TV}{\operatorname{TV}}

\allowdisplaybreaks[1]
\IEEEoverridecommandlockouts

\begin{document}
\title{Joint Probabilistic Shaping and Symbol-Level Precoding for Covert LEO Satellite Downlinks}
\author{Pigi P. Papanikolaou,~\IEEEmembership{Graduate Student Member,~IEEE}, Dimitrios Bozanis,~\IEEEmembership{Graduate Student Member,~IEEE}, \\Sotiris A. Tegos,~\IEEEmembership{Senior Member,~IEEE}, Christos Masouros,~\IEEEmembership{Fellow,~IEEE},\\ and George K. Karagiannidis,~\IEEEmembership{Fellow,~IEEE}

\thanks{P. P. Papanikolaou, D. Bozanis, S. A. Tegos and G. K. Karagiannidis are with the Department of Electrical and Computer Engineering, Aristotle University of Thessaloniki, 54124 Thessaloniki, Greece (e-mails: pigipapa@auth.gr, dimimpoz@ece.auth.gr, tegosoti@auth.gr, geokarag@auth.gr).}
\thanks{C. Masouros is with the Department of Electronic and Electrical Engineering, University College London, London, UK (e-mail: c.masouros@ucl.ac.uk).}

\vspace{-5mm}

}
\maketitle

\begin{abstract}
The wide footprint of low-Earth-orbit (LEO) satellite downlinks exposes the mere occurrence of a transmission to any receiver within it. This paper studies covert multiuser LEO downlinks, where a passive warden at an unknown but fixed ground location accumulates observations as the satellite moves. We jointly optimize the precoding and the occurrence probabilities of the multiuser symbol combinations to provide covertness against such wardens. Each phase-shift keying (PSK) symbol vector is decomposed into a relative pattern and a common rotation, so that shaping the pattern probabilities leaves the symbols of every legitimate user uniformly distributed. We develop probabilistically shaped symbol-level precoding and probabilistically shaped block-level precoding. We derive the exact chi-square divergence of the warden's Gaussian-mixture observation over the complete session and, by bounding the leakage over a continuous location uncertainty region, obtain certificates that guarantee covertness throughout that region. The certified covertness level is minimized under robust constructive-interference, power, and entropy constraints through an alternating optimization whose updates are equivalent to convex conic problems, and the optimized probabilities follow a Maxwell--Boltzmann-type law. Numerical results show that shaping substantially reduces the certified covertness level of both precoding structures, particularly under large location uncertainty, low elevation, and scarce power.
\end{abstract}


\section{Introduction}
\label{Intro}

\IEEEPARstart{L}{ow-Earth-orbit} (LEO) satellite systems are becoming a key component of non-terrestrial networks, extending connectivity beyond terrestrial infrastructure and supporting broadband access, remote monitoring, and Internet-of-Things services \cite{Kodheli2021,Lin2021}. Their integration with cellular networks has motivated extensive work on satellite propagation, mobility, and radio-access procedures \cite{3GPP38811}. However, the wide footprint of a satellite downlink also exposes the transmitted signal to any receiver within it. In applications where the mere occurrence of communication reveals sensitive activity, protecting the information content is not enough, and the transmission itself must remain undetected. Covert communication addresses this requirement by limiting a passive warden's ability to distinguish transmission from silence \cite{Yan2019}.

The detection capability of the warden is determined by the statistics of its observations. In a multiuser LEO downlink, these statistics depend on two design choices. The precoder determines how much of each transmitted signal leaks toward the warden, while the occurrence probabilities of the multiuser symbol combinations determine how often each leakage level is observed. The two choices are coupled, since a symbol combination that requires little power to serve the legitimate users may produce substantial leakage toward an unintended location. The satellite setting adds two more features. The warden remains at one ground position while the satellite moves, so it can combine observations collected under different geometries over the scheduled session. Moreover, the line-of-sight (LoS) channel is largely determined by the satellite trajectory and the ground position, so the leakage can be bounded over a region of candidate warden locations even when the exact location is unknown. A covert design should therefore account jointly for the precoder, the symbol statistics, and the evidence accumulated over the session.

\subsection{State of the Art}

The information-theoretic foundations of covert communication were established through the square-root law for additive white Gaussian noise (AWGN) channels \cite{Bash2013}. Subsequent works characterized covert communication through channel resolvability \cite{Bloch2016} and derived the fundamental scaling constants for discrete memoryless and AWGN channels \cite{Wang2016}. These results show that covertness is governed by the distinguishability of the warden's observation distributions, and thus by both the statistics of the transmitted signal and the observation duration.

Building on these foundations, practical designs reduce this distinguishability mainly by introducing uncertainty at the warden. The effect of an unknown background-noise power was examined in \cite{Goeckel2016,He2017}, covert transmission aided by an uninformed jammer was studied in \cite{Sobers2017}, and artificial noise (AN) generated by spatially distributed cooperative nodes or by a full-duplex receiver was considered in \cite{Soltani2018,Shahzad2018}. Covertness requirements have also been incorporated into the resource optimization of learning systems \cite{Papanikolaou2025}. In LEO networks, power control and co-channel interference were used to balance covertness and reliability in multi-tier constellations \cite{Feng2024}, and a robust rate-splitting and jamming design was developed under uncertainty in both the legitimate and warden channels \cite{Jia2025}. Satellite uplinks with uncertain transmitter and UAV-warden locations were investigated in \cite{Peng2026}, accounting for detection opportunities over multiple message fragments, while orthogonal time frequency space (OTFS) modulation was studied under joint time- and delay--Doppler-domain detection \cite{GuoOTFS2026}. These works control power, interference, and additional radiation, but treat the statistics of the information-bearing symbols as given.

A finer control of the transmitted signal is offered by symbol-level precoding (SLP), which designs the transmit vector based on both the channel state information (CSI) and the instantaneous data symbols. Constructive-interference (CI) designs exploit known multiuser interference to enlarge the decision margins of the legitimate users \cite{Masouros2015,Alodeh2015,Haqiqatnejad2018,Li2020}, and SLP has been applied to satellite downlinks under per-antenna power constraints \cite{Spano2019}. Symbol-level control has also been exploited for security. CI-based secure precoding was developed in \cite{Khandaker2018}, while joint CI and destructive-interference (DI) designs which strengthen the signals of the legitimate users and degrade those of potential eavesdroppers were proposed for secure integrated sensing and communication (ISAC) under target-angle uncertainty \cite{Boz} and in the near-field regime \cite{boz2}. Moreover, an SLP design for integrated sensing and covert communication was presented in \cite{WangSLP2026}. These works establish that SLP can shape what an unintended receiver observes. However, they optimize the transmit vector of each symbol combination, while the occurrence probabilities of these combinations remain fixed by the data.

The occurrence probabilities themselves can be controlled through probabilistic shaping. Probabilistic amplitude shaping combines nonuniform symbol distributions with coded modulation \cite{Bocherer2015}, and distribution matching provides an invertible mapping from data bits to sequences with a prescribed composition \cite{Schulte2016}. For covert communication, constellation probabilities and transmit power were jointly optimized to improve the covert rate of a single-antenna link \cite{Ma2022}. However, in these works, shaping is applied to the symbols of a single link and therefore changes the symbol distribution observed by the legitimate receiver.

\subsection{Motivation and Contribution}

Despite these advances, important gaps remain for covert multiuser LEO downlinks. First, precoding-based and shaping-based covert designs have been developed separately. SLP designs optimize the precoding for given symbol statistics, whereas shaping designs modify the symbol distribution of a single link without exploiting the multiuser spatial structure. However, in a multiuser downlink, the joint symbol combinations across users can be shaped, while the symbol distribution of each individual user remains unchanged. To the best of our knowledge, this possibility has not been explored for covert communication, and the coupling between precoding and pattern probabilities has not been captured. Furthermore, with discrete symbols and symbol-dependent precoding, the warden observes a Gaussian mixture whose weights are the pattern probabilities, and whose distinguishability accumulates over blocks with different geometries. To the best of our knowledge, an exact characterization of this finite-session observation law is not available. In addition, the warden's location is usually assumed known, modeled statistically, or represented by a finite set of candidate points, which provides no guarantee between the considered locations. For a fixed warden observing the entire scheduled session, the guarantee must instead hold at every position of a continuous region, after accumulating the evidence of all blocks at that position.

Motivated by the above, this paper develops a joint probabilistic-shaping and precoding framework for covert multiuser LEO downlinks. Each  PSK symbol vector is decomposed into a relative multiuser pattern and a common rotation, so that the pattern probabilities can be shaped while every legitimate user still observes uniformly distributed PSK symbols. The precoding vectors and the pattern probabilities are then jointly designed to minimize a certified bound on the warden's finite-session distinguishability, valid over the entire location uncertainty region, subject to robust CI, power, and entropy requirements. The main contributions are summarized as follows:

\begin{itemize}
\item We propose two covert designs that jointly optimize the precoding and the probabilities of relative multiuser patterns: probabilistically shaped SLP (PS-SLP), with one precoding vector per pattern, and probabilistically shaped block-level precoding (PS-BLP), with a common linear precoder per block. Both preserve uniform PSK marginals at every legitimate user and allocate an input-entropy requirement across the blocks of the session. The CI requirements are guaranteed for every error within a norm-bounded set around the predicted LoS channels of the legitimate users, which accounts for orbit and position errors, imperfect phase and Doppler compensation, and atmospheric effects.

\item We derive the exact chi-square divergence of the warden's weighted Gaussian-mixture observation and extend it to the complete session, for a warden that knows the precoding vectors and the pattern distributions. Since the location of the warden is unknown, we bound its leakage over a continuous location uncertainty region, using only a lower bound on its noise variance and an upper bound on its antenna gain, and obtain finite-session certificates that hold at every position of this region. The tightest certificate is never looser than a simpler alternative bound, and, under exact envelopes, its leading low-leakage term coincides with that of the exact divergence.

\item We develop an alternating optimization in which both the precoding and the probability updates are equivalent to convex conic problems, so that each update is solved exactly. We further show that the optimized probabilities follow an exponential law in a per-pattern leakage-and-power cost, analogous to Maxwell--Boltzmann shaping, which explains how the session entropy is allocated across blocks.

\item Numerical results show that shaping substantially reduces the certified covertness level of both precoding structures, and that its benefit grows with the warden-location uncertainty, at low satellite elevation, and under scarce power. Under larger channel errors at the legitimate users, shaping also extends the range of QoS requirements for which covertness can be certified. The results also quantify the conservatism of the certificates and the session length that can be certified for a target covertness level.
\end{itemize}

\subsection{Structure and Notation}

The remainder of this paper is organized as follows. Section~\ref{sec:model} introduces the system model and the relative-pattern signaling. Section~\ref{sec:covertness} derives the finite-session covertness analysis and the robust certificates. Section~\ref{sec:design} presents the joint design problem and its solution. Section~\ref{Num} reports the numerical results, and Section~\ref{Conc} concludes the paper.

\emph{Notation:} Boldface lowercase and uppercase letters denote vectors and matrices, respectively. The superscripts $(\cdot)^T$, $(\cdot)^H$, and $(\cdot)^*$ denote transpose, Hermitian transpose, and complex conjugation. The sets of real and complex numbers are denoted by $\mathbb R$ and $\mathbb C$, and $\mathrm j=\sqrt{-1}$. The operators $\Re\{\cdot\}$ and $\Im\{\cdot\}$ extract the real and imaginary parts, and $|\cdot|$ and $\|\cdot\|$ denote the scalar modulus and the Euclidean norm. The symbol $\Pr(\cdot)$ denotes probability, and $\mathcal{CN}(0,\sigma^2)$ denotes a zero-mean proper complex Gaussian distribution with variance $\sigma^2$. Natural, base-two, and base-ten logarithms are denoted by $\ln$, $\log_2$, and $\log$, respectively. The entropy $H(\cdot)$ is measured in bits, and $\operatorname{Var}_{\boldsymbol\pi}(\cdot)$ denotes the variance under the pattern distribution $\boldsymbol\pi$. Finally, $\chi^2$ denotes the chi-square divergence and $\operatorname{TV}(\mathbb P_1,\mathbb P_0)$ the total-variation distance between the probability laws $\mathbb P_1$ and $\mathbb P_0$.

\section{System Model and Relative-Pattern Shaping}
\label{sec:model}

\begin{figure}
    \centering
    \includegraphics[width=0.9\columnwidth]{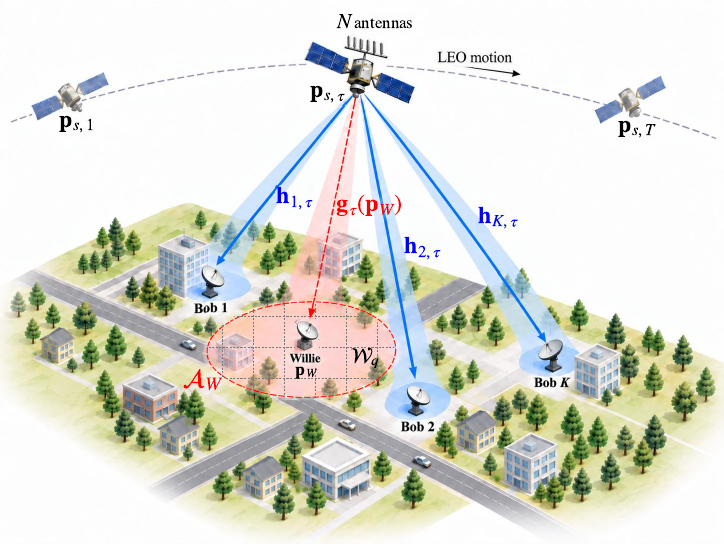}
    \vspace{-2mm}
    \caption{LEO satellite downlink with multiple Bobs and a passive Willie.}
    \vspace{-4mm}
    \label{fig:sys_model}
\end{figure}

\subsection{LEO Geometry and Channel Uncertainty}
\label{subsec:geometry}


As illustrated in Fig.~\ref{fig:sys_model}, we consider a LEO satellite equipped 
with a uniform linear array (ULA) comprising $N$ elements with half-wavelength spacing and a dedicated RF chain for each antenna element, serving $K$ single-antenna legitimate users, referred to as Bobs, indexed by $k\in\mathcal K\triangleq\{1,\ldots,K\}$. A passive single-antenna warden, Willie (W), attempts to distinguish transmission from silence during the scheduled downlink session. The antenna elements and their corresponding RF chains are indexed by $n\in\mathcal N\triangleq\{1,\ldots,N\}$. The scheduled session is partitioned into $T$ quasi-static transmission blocks, indexed by $\tau\in\mathcal T\triangleq\{1,\ldots,T\}$, each comprising $L$ channel uses, for a total of $n_{\mathrm{obs}}=TL$ channel uses. The symbol duration is $T_s$, so each transmission block spans $LT_s$, and the total active transmission time is $n_{\mathrm{obs}}T_s$. The block duration is chosen so that the effective channels can be treated as constant after timing, carrier-phase, and Doppler compensation. The satellite trajectory and array orientation over the scheduled transmission blocks are assumed known when the precoding vectors are designed. 

Let $\mathbf p_{s,\tau}$ denote the satellite position at transmission block $\tau$ and $\mathbf p_u$ the position of a ground node $u\in\mathcal K\cup\{W\}$. The slant range, the unit line-of-sight (LoS) vector, and the direction cosine with respect to the array axis are
\begin{equation}
\begin{aligned}
d_{u,\tau}&=\norm{\mathbf p_u-\mathbf p_{s,\tau}},\\
\mathbf q_{u,\tau}&=\frac{\mathbf p_u-\mathbf p_{s,\tau}}{d_{u,\tau}},\\
\nu_{u,\tau}&=\mathbf o_{a,\tau}^{T}\mathbf q_{u,\tau},
\end{aligned}
\label{eq:geometry}
\end{equation}
where $\mathbf o_{a,\tau}$ is the unit vector along the array axis. Hence, for a generic direction cosine $\nu\in[-1,1]$, the normalized
ULA response is
\begin{equation}
\mathbf a(\nu)=\frac{1}{\sqrt{N}}
[1,e^{-\mathrm j\pi\nu},\ldots,e^{-\mathrm j\pi(N-1)\nu}]^T.
\label{eq:steering}
\end{equation}
Under the adopted free-space LoS model, the channel amplitude toward node $u$ is
\begin{equation}
\beta_{u,\tau}=\frac{\lambda}{4\pi d_{u,\tau}}
\sqrt{NG_sG_u}\,10^{-\Lambda_{u,\tau}/20},
\label{eq:beta}
\end{equation}
where $G_s$ and $G_u$ denote element and receive-antenna power gains and $\Lambda_{u,\tau}$ collects additional propagation losses in dB. 
The satellite predicts each Bob's channel from the known satellite trajectory and the Bob's reported position. Residual orbit and position errors, imperfect phase and Doppler compensation, and atmospheric propagation effects make the actual channel deviate from this prediction. We therefore model the nominal and actual channels of Bob $k$ as
\begin{align}
\widehat{\mathbf h}_{k,\tau}
&=\beta_{k,\tau}e^{-\mathrm j2\pi d_{k,\tau}/\lambda}
\mathbf a(\nu_{k,\tau}),\label{eq:bob_nom}\\
\mathbf h_{k,\tau}&=\widehat{\mathbf h}_{k,\tau}+\mathbf e_{k,\tau},
\quad \norm{\mathbf e_{k,\tau}}\leq\delta^B_{k,\tau},
\label{eq:bob_unc}
\end{align}
where $\mathbf e_{k,\tau}\in\mathbb C^N$ denotes the channel prediction error of Bob $k$ in transmission block $\tau$, and $\delta^B_{k,\tau}=\kappa_B \|\widehat{\mathbf h}_{k,\tau}\|$ specifies its norm bound.

Willie's exact position $\mathbf p_W$ is unknown to the satellite, which only knows that it lies in a bounded ground region $\mathcal A_W$, referred to as the location uncertainty region. Willie remains at the same position throughout the session, but his channel varies across transmission blocks as the satellite moves, according to
\begin{equation}
\mathbf g_\tau(\mathbf p_W)
=\beta_{W,\tau}(\mathbf p_W)
e^{-\mathrm j2\pi d_{W,\tau}(\mathbf p_W)/\lambda}
\mathbf a\bigl(\nu_{W,\tau}(\mathbf p_W)\bigr).
\label{eq:willie_channel}
\end{equation}
Regarding Willie, the satellite knows only the location uncertainty region $\mathcal A_W$, an upper bound on his receive-antenna gain $G_W$, and a lower bound $\sigma_{w,\min}^2$ on his receiver-noise variance, such that $\sigma_w^2\geq\sigma_{w,\min}^2>0$.

\subsection{Relative-Pattern Signaling}
\label{subsec:patterns}

We consider $\mathcal M$-PSK, with modulation order $\mathcal M$ divisible by four, and alphabet $\mathcal S =\left\{e^{\mathrm j2\pi m/\mathcal M}: m=0,\ldots,\mathcal M-1\right\}.$ The alphabet is closed under multiplication and conjugation, which enables the relative-pattern representation introduced next. The restriction to $\mathcal M$ divisible by four ensures that $\mathcal S$ is a union of rotated QPSK subsets.
At channel use $\ell$ of transmission block $\tau$, the multiuser symbol vector is
\begin{equation}
\mathbf s_{\tau,\ell}
=
[s_{1,\tau,\ell},\ldots,s_{K,\tau,\ell}]^T
\in\mathcal S^K,
\end{equation}
where $s_{k,\tau,\ell}$ is the $\mathcal{M}$-PSK symbol intended for Bob $k$.

To parameterize the joint symbol distribution, we use Bob~1 as a phase reference and describe the remaining Bobs' symbols relative to it. Accordingly, define the set of relative PSK patterns as
\begin{equation}
\mathcal U
=
\left\{
\begin{array}{l}
\mathbf u_r=[1,u_{2,r},\ldots,u_{K,r}]^T :\\
u_{k,r}\in\mathcal S,\ k=2,\ldots,K
\end{array}
\right\}_{r=1}^{R},
\label{eq:relative_set}
\end{equation}
which contains $R=\mathcal{M}^{K-1}$ patterns. We enumerate them as $\mathbf u_1,\ldots,\mathbf u_R$, indexed by $r\in\mathcal R\triangleq\{1,\ldots,R\}$, and denote by $u_{k,r}$ the $k$-th entry of $\mathbf u_r$.
Each multiuser symbol vector admits the unique decomposition
\begin{equation}
\mathbf s_{\tau,\ell}
=\zeta_{\tau,\ell}\mathbf u_{\mathsf r_{\tau,\ell}},
\label{eq:relative_decomp}
\end{equation}
where $\zeta_{\tau,\ell}=s_{1,\tau,\ell}$ is the common rotation and $\mathsf r_{\tau,\ell}\in\mathcal R$ is the corresponding relative-pattern index.

\subsection{Probabilistic Multiuser Pattern Shaping}
\label{subsec:shaping}
The pattern probabilities $\boldsymbol\pi_\tau=[\pi_{\tau,1},\ldots,\pi_{\tau,R}]^T$ are design variables optimized jointly with the precoding vectors. In transmission block $\tau$, they specify
\begin{equation}
\Pr(\mathsf r_{\tau,\ell}=r)=\pi_{\tau,r},
\quad
\sum_{r=1}^{R}\pi_{\tau,r}=1.
\label{eq:pi_def}
\end{equation}
We optimize the joint symbol distribution within a jointly rotation-invariant family, generated by a common rotation $\zeta_{\tau,\ell}$ that is uniform over $\mathcal S$ and independent of the relative-pattern index $\mathsf r_{\tau,\ell}$. This construction guarantees uniform symbol marginals at every Bob while restricting the probability design to the relative patterns. The relative-pattern indices and common rotations are drawn independently across channel uses, with $\boldsymbol\pi_\tau$ fixed during transmission block $\tau$.
For each relative pattern $r$ at transmission block $\tau$, the satellite designs a representative SLP vector $\mathbf v_{\tau,r}=[v_{1,\tau,r},\ldots,v_{N,\tau,r}]^T \in\mathbb C^N$, where $n\in\mathcal N$ and $v_{n,\tau,r}$ denotes the signal transmitted by the $n$-th RF chain. Hence, the transmitted signal is
\begin{equation}
\mathbf x_{\tau,\ell}
=\zeta_{\tau,\ell}\mathbf v_{\tau,\mathsf r_{\tau,\ell}}.
\label{eq:transmitted}
\end{equation}
Thus, the satellite designs only $R$ representative precoding vectors $\mathbf v_{\tau,1},\ldots,\mathbf v_{\tau,R}$ per transmission block, one for each relative pattern, rather than one for each of the $\mathcal{M}^K$ multiuser symbol vectors. The common rotation $\zeta_{\tau,\ell}$ rotates the selected precoding vector and is Bob~1's $\mathcal{M}$-PSK symbol, rather than artificial randomness. This reduction uses the rotational symmetry of PSK CI precoding~\cite{Masouros2015}.

\begin{lemma}
\label{lem:entropy}
For any pattern distribution $\boldsymbol\pi_\tau$ and any channel use $\ell$ of transmission block $\tau$, the symbol $s_{k,\tau,\ell}$ of every Bob is uniformly distributed over $\mathcal S$, and
\begin{equation}
\begin{aligned}
H(\mathbf s_{\tau,\ell})
&=\log_2 \mathcal{M}+H(\boldsymbol\pi_\tau),\\
H(\boldsymbol\pi_\tau)
&=-\sum_{r=1}^{R}\pi_{\tau,r}\log_2\pi_{\tau,r}.
\end{aligned}
\label{eq:entropy}
\end{equation}
\end{lemma}

\begin{proof}
For a fixed relative pattern $r$, multiplication by $u_{k,r}\in\mathcal S$ permutes the $\mathcal M$-PSK alphabet. Since $\zeta_{\tau,\ell}$ is uniform over $\mathcal S$ and independent of the pattern index, $s_{k,\tau,\ell}=\zeta_{\tau,\ell}u_{k,r}$ is uniform over $\mathcal S$ conditional on every $r$, and hence for any distribution $\boldsymbol\pi_\tau$. Moreover, the mapping $(\zeta_{\tau,\ell},\mathsf r_{\tau,\ell}) \mapsto\mathbf s_{\tau,\ell}$ is bijective. Independence therefore gives $H(\mathbf s_{\tau,\ell}) =H(\zeta_{\tau,\ell})+H(\mathsf r_{\tau,\ell}) =\log_2 \mathcal{M}+H(\boldsymbol\pi_\tau)$.
\end{proof}

Lemma~\ref{lem:entropy} implies that each Bob can use a standard $\mathcal M$-PSK decision rule without knowing $\boldsymbol\pi_\tau$, while nonuniform shaping reduces the multiuser input entropy from $K\log_2 \mathcal M$ to $\log_2 \mathcal M+H(\boldsymbol\pi_\tau)$ bits per channel use. The symbols of different Bobs are then generally statistically dependent. The satellite generates the multiuser symbol vectors jointly, and \eqref{eq:entropy} describes the entropy of this joint input rather than achievable rates for the individual Bobs.


\subsection{Constructive Interference and Power Constraints}
\label{subsec:ci}

At channel use $\ell$ of transmission block $\tau$, Bob $k$ receives
\begin{equation}
y_{k,\tau,\ell}
=\mathbf h_{k,\tau}^{H}\mathbf x_{\tau,\ell}
+\omega^B_{k,\tau,\ell},
\label{eq:bob_rx}
\end{equation}
where $\omega^B_{k,\tau,\ell}\sim\CN(0,\sigma_k^2)$ denotes receiver noise with power $\sigma_k^2$. The CI design places the noiseless received signal inside the decision region of the intended symbol $s_{k,\tau,\ell}$ with a prescribed margin. Multiplication by the conjugate of the intended symbol maps every PSK symbol to the decision region of $1$, giving the symbol-aligned noiseless signal
\begin{equation}
z_{k,\tau,r}
=s_{k,\tau,\ell}^{*}
\mathbf h_{k,\tau}^{H}\mathbf x_{\tau,\ell}\\
=\mathbf h_{k,\tau}^{H}
\mathbf v_{\tau,r}u_{k,r}^{*},
\label{eq:aligned}
\end{equation}
for any channel use satisfying $\mathsf r_{\tau,\ell}=r$. The second equality follows from $s_{k,\tau,\ell}=\zeta_{\tau,\ell}u_{k,r}$ and $|\zeta_{\tau,\ell}|=1$. Hence, $z_{k,\tau,r}$ does not depend on the common rotation, and $R$ pattern-wise CI conditions per Bob and transmission block suffice.  

For $\mathcal M$-PSK, the decision region of the rotated symbol is a sector with half-angle $\phi=\pi/\mathcal M$ around the positive real axis. Given a prescribed SNR threshold $\gamma_{\rm req}>0$, the CI condition \cite{Masouros2015,Alodeh2015} requires $z_{k,\tau,r}$ to lie in this sector shifted by $\sqrt{\gamma_{\rm req}}\sigma_k$ along the positive real axis, i.e.,
\begin{equation}
|\Im\{z_{k,\tau,r}\}|
\leq
\left(\Re\{z_{k,\tau,r}\}
-\sqrt{\gamma_{\rm req}}\sigma_k\right)\tan\phi.
\label{eq:nominal_ci}
\end{equation}
This shift keeps the noiseless received signal away from both original decision boundaries, providing a prescribed noise margin.
Since the actual channel is unknown, \eqref{eq:nominal_ci} must hold for every channel error in \eqref{eq:bob_unc}. Splitting the absolute value into two linear inequalities, indexed by $b\in\{-1,+1\}$, the worst-case error in each of them follows from the Cauchy--Schwarz inequality, which yields the equivalent condition
\begin{align}
&b\,\Im\{\widehat z_{k,\tau,r}\}
-\tan\phi\,\Re\{\widehat z_{k,\tau,r}\}
+\delta^B_{k,\tau}\sec\phi\,\|\mathbf v_{\tau,r}\|
\nonumber\\
&\hspace{15mm}
+\sqrt{\gamma_{\rm req}}\sigma_k\tan\phi\leq 0,
\quad b\in\{-1,+1\},
\label{eq:robust_ci}
\end{align}
where $\widehat z_{k,\tau,r}=\widehat{\mathbf h}_{k,\tau}^{H}\mathbf v_{\tau,r}u_{k,r}^{*}$ is the aligned signal for the nominal channel. The term $\delta^B_{k,\tau}\sec\phi\,\|\mathbf v_{\tau,r}\|$ is the margin lost compared to the worst-case channel error.

Since the common rotation has unit magnitude, the instantaneous power of RF chain $n$ is $|v_{n,\tau,r}|^2$ and the total instantaneous transmit power is $\|\mathbf v_{\tau,r}\|^2$. We impose a peak limit $P_n^{\rm pk}$ per RF chain, which reflects amplifier saturation, an average limit $P_{n,\tau}^{\rm av}$ per RF chain and block, and an energy budget $E_{\rm sess}$ for the complete session, which reflects the energy available to the satellite during the scheduled session:
\begin{subequations}\label{eq:power}
\begin{align}
|v_{n,\tau,r}|^2&\leq P_n^{\rm pk}, &&\forall n,\tau,r, \label{eq:peak}\\
\sum_{r=1}^{R}\pi_{\tau,r}|v_{n,\tau,r}|^2&\leq P_{n,\tau}^{\rm av}, &&\forall n,\tau, \label{eq:avgfeed}\\
LT_s\sum_{\tau=1}^{T}\sum_{r=1}^{R}\pi_{\tau,r}\|\mathbf v_{\tau,r}\|^2&\leq E_{\rm sess}. \label{eq:energy}
\end{align}
\end{subequations}
The peak constraint applies to every precoding vector, whereas \eqref{eq:avgfeed} and \eqref{eq:energy} constrain averages over the pattern distribution, since the transmitted power depends on the selected pattern.

\section{Finite-Session Covertness Analysis and Robust Certificates}
\label{sec:covertness}

This section quantifies how well Willie can detect the transmission. We first derive the exact distinguishability of transmission and silence for a given Willie position in Section~\ref{subsec:mixture}. Since this position is unknown, we then bound Willie's channel over the entire uncertainty region in Section~\ref{subsec:location} and combine the two results into a covertness certificate that holds for every admissible position in Section~\ref{subsec:certificates}. Finally, Section~\ref{subsec:comparison} interprets the certificates in the low-leakage regime and compares them with a simpler bound.

\subsection{Weighted Gaussian-Mixture Observation}
\label{subsec:mixture}

Willie must decide whether the satellite transmits or remains silent, based on all his observations during the session. Through his channel \eqref{eq:willie_channel}, his received signal at channel use $\ell$ of block $\tau$ is
\begin{equation}
Y_{\tau,\ell}
=\mathbf g_\tau^H(\mathbf p_W)\mathbf x_{\tau,\ell}
+\omega_{\tau,\ell}, 
\label{eq:willie_rx}
\end{equation}
where $\omega_{\tau,\ell}\sim\CN(0,\sigma_w^2)$ denotes the receiver noise, independent across channel uses.
For relative pattern $r$, define the leakage coefficient $ \mu_{\tau,r}(\mathbf p_W)=\mathbf g_\tau^H(\mathbf p_W)\mathbf v_{\tau,r}$,
i.e., the transmitted signal component received by Willie before the common rotation. Here, leakage refers to the received signal component relevant to detection, rather than to information leakage as in secrecy.
Substituting \eqref{eq:transmitted} into \eqref{eq:willie_rx} and suppressing the channel-use index, Willie's binary detection problem in block $\tau$ is 
\begin{equation}
\mathcal H_0:Y=\omega,\qquad
\mathcal H_1:Y=\zeta\mu_{\tau,\mathsf r}+\omega,
\label{eq:hypotheses}
\end{equation}
where $\mathcal H_0$ and $\mathcal H_1$ represent silence and transmission, respectively.

Willie is assumed to know his channel, his noise variance, the scheduled transmission blocks, the precoding vectors $\{\mathbf v_{\tau,r}\}$, and the distributions $\{\boldsymbol\pi_\tau\}$. Only the instantaneous realizations of $\mathsf r$ and $\zeta$ are unknown to him. The precoding table acts as a codebook known to Willie, analogous to a known constellation, whereas the selected entry and its common rotation depend on the instantaneous symbols, which remain unknown to him. Their distributions are specified in Section~\ref{subsec:shaping}. In line with Kerckhoffs's principle, this is a worst-case assumption for the transmitter, since the whole system design is known to the adversary. Therefore, the resulting guarantee also holds for a less-informed Willie, because additional side information cannot degrade the performance of his optimal detector.

Let $p_a(y)$ denote the complex Gaussian density with mean $a$ and variance $\sigma_w^2$, namely
\begin{equation}
p_a(y)=\frac{1}{\pi\sigma_w^2}
\exp\left(-\frac{|y-a|^2}{\sigma_w^2}\right).
\label{eq:gaussian_density}
\end{equation}
Under $\mathcal H_0$, the observation density is $p_0(y)\equiv p_a(y)|_{a=0}$.
Under $\mathcal H_1$, averaging over the unknown relative pattern and common PSK rotation gives the weighted Gaussian mixture
\begin{equation}
p_{1,\tau}(y)
=\sum_{r=1}^{R}\frac{\pi_{\tau,r}}{\mathcal{M}}
\sum_{\zeta\in\mathcal S}
p_{\zeta\mu_{\tau,r}}(y).
\label{eq:mixture}
\end{equation}
The position $\mathbf p_W$ enters $p_{1,\tau}(y)$ through the leakage coefficients $\mu_{\tau,r}(\mathbf p_W)$, although we suppress this argument in the density notation. Under the assumed spatially constant receiver-noise variance, $p_0(y)$ does not depend on $\mathbf p_W$.

Willie's minimum detection error is determined by the total-variation distance between $p_{1,\tau}$ and $p_0$, which, like the Kullback--Leibler divergence, has no closed form for a Gaussian mixture. We therefore use the chi-square divergence
\begin{equation}
\chi_\tau^2(\mathbf p_W)
=\int_{\mathbb C}
\frac{[p_{1,\tau}(y)-p_0(y)]^2}{p_0(y)}
\,dy,
\label{eq:chi}
\end{equation}
which admits an exact expression for the considered mixture, factorizes over independent channel uses, and upper-bounds the total-variation distance \cite{Tsybakov2009}. 
The following proposition provides its exact expression for the adopted PSK shaping model. It is stated in terms of $D_\tau=1+\chi_\tau^2$, which is the quantity that multiplies over independent channel uses.

\begin{proposition}
\label{prop:exact}
Let
$D_\tau(\mathbf p_W)=1+\chi_\tau^2(\mathbf p_W)$.
Then
\begin{equation}
D_\tau(\mathbf p_W)
=\sum_{r=1}^{R}\sum_{t=1}^{R}
\pi_{\tau,r}\pi_{\tau,t}
\Phi_{\mathcal M}\left(
\frac{\mu_{\tau,r}^*\mu_{\tau,t}}{\sigma_w^2}
\right),
\label{eq:exactD}
\end{equation}
where
\begin{equation}
\Phi_{\mathcal M}(z)
=\frac{1}{\mathcal M}
\sum_{m=0}^{\mathcal M-1}
\exp\!\left(
2\Re\{e^{\mathrm j2\pi m/\mathcal M}z\}
\right).
\label{eq:phi4}
\end{equation}
\end{proposition}
\begin{proof} The proof is presented in Appendix~\ref{app:exact}.\end{proof}

For the geometric LoS channel in \eqref{eq:willie_channel}, the propagation phase is common to all patterns within a block and cancels in every product $\mu_{\tau,r}^{*}\mu_{\tau,t}$, so $D_\tau(\mathbf p_W)$ is independent of it. The kernel $\Phi_{\mathcal M}$ averages $e^{2\Re\{\upsilon z\}}$ over the $\mathcal M$ PSK rotations $\upsilon\in\mathcal S$, thus explains Willie's ignorance of the common rotation, while \eqref{eq:exactD} depends on the precoding vectors and on the probabilities through every pair of patterns $(r,t)$ through the products $\mu_{\tau,r}^*\mu_{\tau,t}$.

Next, we extend the analysis to the complete session. Let $\chi_{\rm sess}^2(\mathbf p_W)$ denote the chi-square divergence between the joint distributions of Willie's session observations under $\mathcal H_1$ and $\mathcal H_0$. Since the pattern indices, the common rotations, and the noise samples are independent across channel uses, it factorizes over the channel uses. The $L$ uses of block $\tau$ are identically distributed and each contributes the factor $D_\tau(\mathbf p_W)$, thus
\begin{equation}
1+\chi_{\rm sess}^2(\mathbf p_W)
=\prod_{\tau=1}^{T}
D_\tau(\mathbf p_W)^L.
\label{eq:session_product}
\end{equation}
Taking the logarithm yields
\begin{equation}
\ln\left(1+\chi_{\rm sess}^2(\mathbf p_W)\right)
=L\sum_{\tau=1}^{T}
\ln D_\tau(\mathbf p_W).
\label{eq:session_log}
\end{equation}
Since Willie's channel changes between blocks while his position remains fixed, \eqref{eq:session_log} accumulates detection evidence over different LEO geometries, with each block contributing according to its own leakage. Silent intervals between blocks do not contribute additional distinguishability, since the corresponding observations follow the same noise-only distribution under both hypotheses.

Let $\mathbb P_0$ and $\mathbb P_1(\mathbf p_W)$ denote the joint distributions of Willie's session observations under $\mathcal H_0$ and $\mathcal H_1$, respectively. For a fixed $\mathbf p_W$, the minimum sum $\xi^\star(\mathbf p_W)$ of false-alarm and missed-detection probabilities over all detectors satisfies \cite{Bash2013,Tsybakov2009}
\begin{equation}
\begin{aligned}
\xi^\star(\mathbf p_W)
&=1-\TV\!\left(\mathbb P_1(\mathbf p_W),\mathbb P_0\right)\\
&\geq 1-\frac12\sqrt{\chi_{\rm sess}^2(\mathbf p_W)}.
\end{aligned}
\label{eq:TV}
\end{equation}
A blind decision achieves a false-alarm and missed-detection probability sum of one.
Following the standard covertness criterion \cite{Bash2013,Yan2019}, the transmission is called $\epsilon$-covert, for a prescribed $0<\epsilon<1$, if $\xi^\star(\mathbf p_W)\geq1-\epsilon$ for every $\mathbf p_W\in\mathcal A_W$.
Covertness therefore requires Willie's optimal detector to perform close to a blind decision, rather than requiring zero received signal energy.
Since \eqref{eq:session_log} is additive over the blocks, we express guarantees through bounds on $\ln(1+\chi_{\rm sess}^2)$. Any nonnegative $C$ satisfying
\begin{equation}
\ln\left(1+\chi_{\rm sess}^2(\mathbf p_W)\right)\leq C,\quad\forall\mathbf p_W\in\mathcal A_W,
\end{equation}
guarantees
\begin{equation}
\xi^\star(\mathbf p_W)\geq\max\left\{0,1-\frac12\sqrt{e^C-1}\right\},\quad\forall\,\mathbf p_W\in\mathcal A_W.
\label{eq:general_covertness}
\end{equation}
By \eqref{eq:general_covertness}, requiring $\frac12\sqrt{e^C-1}\leq\epsilon$ guarantees $\epsilon$-covertness, equivalently,
\begin{equation}
C\leq C_\epsilon\triangleq\ln(1+4\epsilon^2).
\label{eq:target_covertness}
\end{equation}
Next, we bound the leakage over the continuous location region.

\subsection{Leakage Envelope over Willie's Uncertainty Region}
\label{subsec:location}

Willie's position enters \eqref{eq:exactD} only through the leakage coefficients $\mu_{\tau,r}(\mathbf p_W)$. To account for the unknown position, we therefore seek upper bounds on $|\mu_{\tau,r}(\mathbf p_W)|$ that hold throughout $\mathcal A_W$. A single bound over the whole region would be loose, since Willie's channel varies considerably across a region of tens of kilometers. We therefore partition $\mathcal A_W$ into $Q$ cells $\{\mathcal W_q\}_{q\in\mathcal Q}$, with $\mathcal Q=\{1,\ldots,Q\}$ and $\bigcup_{q\in\mathcal Q}\mathcal W_q=\mathcal A_W$, and bound the channel variation within each cell. Each cell is enclosed by a ball of radius $r_q$ centered at $\mathbf p_q$, i.e.,
\begin{equation}
\|\mathbf p_W-\mathbf p_q\|\leq r_q,
\quad \forall\mathbf p_W\in\mathcal W_q.
\end{equation}
The cells partition $\mathcal A_W$, while their enclosing balls can overlap. The balls are used only to bound channel variation within each cell, thus their overlap does not affect the guarantee.

For transmission block $\tau$, let $d_{\tau,q}$, $\nu_{\tau,q}$, and $\beta_{\tau,q}$ denote the slant range, the direction cosine, and the channel amplitude at the cell center $\mathbf p_q$, where $\beta_{\tau,q}$ is computed from \eqref{eq:beta} using the upper bound on $G_W$.
Since $|\mu_{\tau,r}(\mathbf p_W)|=\beta_{W,\tau}(\mathbf p_W)\,|\widetilde{\mathbf a}^H(\nu_{W,\tau}(\mathbf p_W))\mathbf v_{\tau,r}|$, Willie's position affects the leakage through two factors, the channel amplitude and the array response, which we bound in turn.
From the reverse triangle inequality, the slant range of every location in $\mathcal W_q$ satisfies $d_{W,\tau}(\mathbf p_W)\geq d_{\tau,q}-r_q$. Since $\beta$ is inversely proportional to the slant range in \eqref{eq:beta}, and assuming $r_q<d_{\tau,q}$ and worst-case antenna gains and losses within the cell, the channel amplitude is bounded by
\begin{equation}
\beta_{\tau,q}^{\max}=\beta_{\tau,q}\frac{d_{\tau,q}}{d_{\tau,q}-r_q}.
\label{eq:beta_max}
\end{equation}

Next, we bound the variation of the array response. Intuitively, a displacement of at most $r_q$, observed from a distance of at least $d_{\tau,q}-r_q$, can change the LoS direction by at most $r_q/(d_{\tau,q}-r_q)$. Formally, the unit LoS vector $\mathbf q(\mathbf p)=(\mathbf p-\mathbf p_{s,\tau})/\|\mathbf p-\mathbf p_{s,\tau}\|$ in \eqref{eq:geometry} has a Jacobian of norm $1/\|\mathbf p-\mathbf p_{s,\tau}\|\leq1/(d_{\tau,q}-r_q)$ within the enclosing ball. Since $\nu$ is the projection of $\mathbf q$ on the unit array axis, every direction cosine in $\mathcal W_q$ lies within
\begin{equation}
\Delta\nu_{\tau,q}=\min\left\{2,\frac{r_q}{d_{\tau,q}-r_q}\right\}
\label{eq:delta_nu}
\end{equation}
of $\nu_{\tau,q}$, i.e., in the interval $[\nu_{\tau,q}^-,\nu_{\tau,q}^+]$ with $\nu_{\tau,q}^-=\max\{-1,\nu_{\tau,q}-\Delta\nu_{\tau,q}\}$ and $\nu_{\tau,q}^+=\min\{1,\nu_{\tau,q}+\Delta\nu_{\tau,q}\}$. To translate this interval into a bound on the array response, we use the phase-centered steering vector
\begin{equation}
\widetilde{\mathbf a}(\nu)=e^{\mathrm j\pi(N-1)\nu/2}\mathbf a(\nu).
\label{eq:centering}
\end{equation}
The centering does not change $|\widetilde{\mathbf a}^H(\nu)\mathbf v|$, but it minimizes the norm of the derivative, which becomes $\pi\sqrt{(N^2-1)/12}$, so the resulting bound is tighter. Evaluating the response only at the cell center would require a margin covering the whole interval. Instead, we evaluate it at $I_{\mathrm a}\geq2$ equally spaced anchors
\begin{equation}
\nu_{\tau,q,i_{\mathrm a}}
=\nu_{\tau,q}^{-}
+\frac{i_{\mathrm a}-1}{I_{\mathrm a}-1}
(\nu_{\tau,q}^{+}-\nu_{\tau,q}^{-}),
\quad i_{\mathrm a}=1,\ldots,I_{\mathrm a}.
\end{equation}
so that every point of the interval lies within 
$\Delta\nu_{\tau,q}^{\mathrm{res}}
={(\nu_{\tau,q}^{+}-\nu_{\tau,q}^{-})}/{[2(I_{\mathrm a}-1)]}$
of some anchor. The mean-value inequality then gives
\begin{equation}
\begin{aligned}
&\min_{i_{\mathrm a}=1,\ldots,I_{\mathrm a}}
\left\|
\widetilde{\mathbf a}(\nu)
-\widetilde{\mathbf a}(\nu_{\tau,q,i_{\mathrm a}})
\right\|
\leq\delta_{\tau,q}^{a},\\
&\delta_{\tau,q}^{a}\triangleq
\min\left\{
2,\pi\sqrt{\frac{N^2-1}{12}}\,
\Delta\nu_{\tau,q}^{\mathrm{res}}
\right\},
\end{aligned}
\label{eq:array_diff}
\end{equation}
where the value $2$ follows from the unit norm of the steering vectors. Thus, $\delta_{\tau,q}^{a}$ bounds how much the steering vector can change between any admissible direction and its nearest anchor.

Combining \eqref{eq:beta_max} and \eqref{eq:array_diff} with the triangle and Cauchy--Schwarz inequalities, for every $\mathbf p_W\in\mathcal W_q$,
\begin{equation}
\begin{aligned}
&|\mu_{\tau,r}(\mathbf p_W)|\leq A_{\tau,r,q},\\
&A_{\tau,r,q}\triangleq\beta_{\tau,q}^{\max}
\max_{i_{\mathrm a}=1,\ldots,I_{\mathrm a}}
\left|
\widetilde{\mathbf a}^{H}(\nu_{\tau,q,i_{\mathrm a}})
\mathbf v_{\tau,r}
\right|\\
&\hspace{15mm}
+\beta_{\tau,q}^{\max}\delta_{\tau,q}^{a}
\|\mathbf v_{\tau,r}\|.
\end{aligned}
\label{eq:amplitude_guarantee}
\end{equation}
If Willie's channel includes an additional error that is fixed within each block and has norm at most $\delta^{\rm mod}_{\tau,q}$ within the cell, \eqref{eq:amplitude_guarantee} still holds, by the triangle inequality, after adding $\delta^{\rm mod}_{\tau,q}\|\mathbf v_{\tau,r}\|$ to $A_{\tau,r,q}$. Finally, using the lower bound $\sigma_{w,\min}^2$ on Willie's noise variance, the normalized leakage envelope
\begin{equation}
\rho_{\tau,r,q}=\frac{A_{\tau,r,q}^2}{\sigma_{w,\min}^2}
\label{eq:rho}
\end{equation}
satisfies
\begin{equation}
\frac{|\mu_{\tau,r}(\mathbf p_W)|^2}{\sigma_w^2}\leq\rho_{\tau,r,q},\quad\forall\mathbf p_W\in\mathcal W_q.
\label{eq:rho_guarantee}
\end{equation}

\subsection{Finite-Session Covertness Certificate}
\label{subsec:certificates}
The envelope \eqref{eq:rho_guarantee} bounds the magnitude of each leakage coefficient over a cell, whereas \eqref{eq:exactD} depends on the products $\mu_{\tau,r}^*\mu_{\tau,t}$ of every pair of patterns, including their unknown phases. To combine the two, we bound the kernel $\Phi_{\mathcal M}$ by a function of the magnitudes only. We therefore seek a bound in the form of a product of a term for $r$ and a term for $t$, so that the weighted double sum in \eqref{eq:exactD} factorizes into the square of a single sum, which is linear in the probabilities.


\begin{lemma}
\label{lem:product}
For $t_1,t_2\geq 0$ and $z\in\mathbb C$ satisfying $|z|^2\leq t_1t_2$,
\begin{equation}
0\leq \Phi_{\mathcal M}(z)-1
\leq \sinh(t_1)\sinh(t_2)
\leq \bigl(e^{t_1}-1\bigr)\bigl(e^{t_2}-1\bigr).
\label{eq:product_bound}
\end{equation}
\end{lemma}
\begin{proof} The proof is presented in Appendix~\ref{app:kernel}.\end{proof}

The hyperbolic-sine bound is tighter than the exponential bound and is attained for QPSK when $t_1=t_2$ and $z=t_1$. The two bounds are asymptotically equivalent as $t_1,t_2\to0$.
Lemma~\ref{lem:product} applies directly to every term of \eqref{eq:exactD}. For $\mathbf p_W\in\mathcal W_q$, the envelope \eqref{eq:rho_guarantee} gives $|\mu_{\tau,r}^*\mu_{\tau,t}|^2/\sigma_w^4\leq\rho_{\tau,r,q}\rho_{\tau,t,q}$, which is the condition of the lemma with $t_1=\rho_{\tau,r,q}$ and $t_2=\rho_{\tau,t,q}$. The resulting product bounds factorize the double sum, which motivates the probability-weighted leakage measures
\begin{align}
M_{\tau,q}^{\mathrm{sinh}}&\triangleq\sum_{r=1}^{R}\pi_{\tau,r}\sinh(\rho_{\tau,r,q}),\label{eq:Msinh}\\
M_{\tau,q}&\triangleq\sum_{r=1}^{R}\pi_{\tau,r}\bigl(e^{\rho_{\tau,r,q}}-1\bigr).\label{eq:M}
\end{align}
Both are averages of a per-pattern leakage term over the pattern distribution, and $M_{\tau,q}^{\mathrm{sinh}}\leq M_{\tau,q}$ by Lemma~\ref{lem:product}. The following theorem uses these measures to bound the divergence of each block over every cell and, by accumulating the blocks, provides session certificates that hold throughout the uncertainty region.

\begin{theorem}
\label{thm:exp}
For any precoding vectors and pattern distributions, every block $\tau$, every cell $q$, and every $\mathbf p_W\in\mathcal W_q$,
\begin{equation}
\chi_\tau^2(\mathbf p_W)\leq\left(M_{\tau,q}^{\mathrm{sinh}}\right)^2\leq M_{\tau,q}^2.
\label{eq:exo_one}
\end{equation}
Moreover, define
\begin{align}
C_{\mathrm{sinh}}&\triangleq\max_{q\in\mathcal Q}L\sum_{\tau=1}^{T}\ln\!\left[1+\left(M_{\tau,q}^{\mathrm{sinh}}\right)^2\right], \label{eq:Csinh}\\
C_{\mathrm E}&\triangleq\max_{q\in\mathcal Q}L\sum_{\tau=1}^{T}M_{\tau,q}^2.
\label{eq:CE}
\end{align}
Then, for every $\mathbf p_W\in\mathcal A_W$,
\begin{equation}
\ln\!\left(1+\chi_{\mathrm{sess}}^2(\mathbf p_W)\right)\leq C_{\mathrm{sinh}}\leq C_{\mathrm E}.
\label{eq:CE_chain}
\end{equation}
\end{theorem}
\begin{proof} The proof is presented in Appendix~\ref{app:exp}.\end{proof}

The certificates are valid at arbitrary leakage levels. The maximum over cells is taken after accumulating the transmission blocks because Willie occupies one fixed location throughout the session. We optimize $C_{\mathrm E}$, since $e^{\rho}-1$ admits an exact exponential-cone representation, whereas $\sinh(\rho)$ does not. For small leakage, $e^{\rho}-1$ and $\sinh(\rho)$ are nearly equal, so optimizing $C_{\mathrm E}$ instead of $C_{\mathrm{sinh}}$ has only a minor effect on the resulting design. 
Evaluating $C_{\mathrm{sinh}}\leq C_{\mathrm E}$ for the optimized design provides an additional, tighter post-design covertness guarantee.

\subsection{Low-Leakage Interpretation and Comparison with a Separable Bound}

\label{subsec:comparison}

Covertness requires the leakage received by Willie to be weak, so the regime of practical interest is that of small leakage. In this regime, the exact expression \eqref{eq:exactD} and the proposed certificate admit simple expansions, which reveal how shaping reduces Willie's distinguishability and how tight the certificate is. Fix a transmission block and a Willie location, and let $\varrho_{\tau,r}=|\mu_{\tau,r}(\mathbf p_W)|^2/\sigma_w^2$ denote the SNR at Willie for pattern $r$, with $\varrho_{\tau}^{\max}=\max_r\varrho_{\tau,r}$. For every $\mathcal M$ divisible by four, $\Phi_{\mathcal M}(z)=1+|z|^2+O(|z|^4)$.
Hence, \eqref{eq:exactD} yields
\begin{equation}
\ln D_\tau(\mathbf p_W)
=
\left(\sum_{r=1}^{R}\pi_{\tau,r}\varrho_{\tau,r}\right)^2
+O\!\left((\varrho_{\tau}^{\max})^4\right).
\label{eq:exact_small}
\end{equation}
Thus, in the low-leakage regime, Willie's distinguishability is dominated by the mean SNR at Willie over the pattern distribution, which shaping can reduce by favoring patterns with lower leakage. The phases of the leakage coefficients, and hence the pairwise interactions in \eqref{eq:exactD}, affect only higher-order terms. For the proposed measures, $\sinh(\rho)=\rho+O(\rho^3)$ and $e^\rho-1=\rho+O(\rho^2)$ give
\begin{equation}
\begin{aligned}
\bigl(M_{\tau,q}^{\mathrm{sinh}}\bigr)^2
&=\left(\sum_{r=1}^{R}\pi_{\tau,r}\rho_{\tau,r,q}\right)^2+O\bigl((\rho_{\tau,q}^{\max})^4\bigr),\\
M_{\tau,q}^2
&=\left(\sum_{r=1}^{R}\pi_{\tau,r}\rho_{\tau,r,q}\right)^2+O\bigl((\rho_{\tau,q}^{\max})^3\bigr),
\end{aligned}
\label{eq:mean_small}
\end{equation}
where $\rho_{\tau,q}^{\max}=\max_r\rho_{\tau,r,q}$.
When both the spatial envelope and the noise lower bound are exact, $\rho_{\tau,r,q}=\varrho_{\tau,r}$, so both measures in \eqref{eq:mean_small} reproduce the leading term of \eqref{eq:exact_small}.

A simpler alternative to the proposed certificates bounds the kernel by $\Phi_{\mathcal M}(z)\leq e^{|z|^2}$ and then separates the pattern indices through $ab\leq(a^2+b^2)/2$. The resulting certificate and its relation to $C_{\mathrm{sinh}}$ are given next.

\begin{proposition}
\label{prop:separable}
For any precoding vectors and pattern distributions, every block $\tau$, every cell $q$, and every $\mathbf p_W\in\mathcal W_q$,
\begin{equation}
\ln D_\tau(\mathbf p_W)
\leq\bar d_{\tau,q}
\triangleq
2\ln\!\left[\sum_{r=1}^{R}\pi_{\tau,r}e^{\rho_{\tau,r,q}^2/2}\right].
\label{eq:dbar}
\end{equation}
Define the session certificate
\begin{equation}
C_{\mathrm S}\triangleq\max_{q\in\mathcal Q}L\sum_{\tau=1}^{T}\bar d_{\tau,q}.
\label{eq:CS}
\end{equation}
Then, for every $\mathbf p_W\in\mathcal A_W$,
\begin{equation}
\ln\!\left(1+\chi_{\mathrm{sess}}^2(\mathbf p_W)\right)\leq C_{\mathrm{sinh}}\leq C_{\mathrm S}.
\label{eq:separable_session}
\end{equation}
\end{proposition}
\begin{proof} The proof is presented in Appendix~\ref{app:separable}.\end{proof}

To interpret the difference from the proposed certificates, consider the same low-leakage regime with exact spatial envelopes and noise lower bound. Then,
\begin{equation}
\bar d_{\tau,q}-\ln D_\tau
=\operatorname{Var}_{\boldsymbol\pi_\tau}(\varrho_{\tau,\mathsf r})+O\!\left((\varrho_{\tau}^{\max})^4\right).
\label{eq:variance_gap}
\end{equation}
Thus, the simpler bound incurs a leading penalty for the variation of the leakage across patterns, whereas the proposed measures reproduce the squared mean-leakage term without this penalty. This variation is precisely what shaping exploits, since it assigns different probabilities to patterns with different leakage.
Moreover, $C_{\mathrm{sinh}}$ is never looser than $C_{\mathrm S}$, at any leakage level, by Proposition~\ref{prop:separable}. The certificate $C_{\mathrm E}$, which is used for the optimization, does not share this property. For equal leakages $\rho>0$, $\bar d_{\tau,q}=\rho^2$, whereas $M_{\tau,q}^2=(e^\rho-1)^2>\rho^2$.

\section{Joint Precoding and Shaping Design}

\label{sec:design}
The certificate $C_{\mathrm E}$ of Section~\ref{sec:covertness} is an explicit function of the precoding vectors and the pattern distributions, and can therefore serve directly as a design objective. Section~\ref{subsec:joint} formulates its minimization under communication, power, and entropy constraints, and Section~\ref{subsec:ao} develops an alternating optimization with convex subproblems. Section~\ref{subsec:convergence} discusses initialization, convergence, and complexity, while Section~\ref{subsec:benchmarks} presents the benchmark designs.

\subsection{Session Entropy and Problem Formulation}
\label{subsec:joint}
The precoding tables and pattern distributions are designed for the scheduled transmission blocks. During each block, the probabilities remain fixed, while the selected precoding vector changes with the instantaneous relative pattern.
Minimizing the leakage alone would concentrate the probability on the few patterns with the lowest leakage, and the multiuser input would then carry little information. Therefore, we prescribe an average relative-pattern entropy $\bar H$ over the session, which by Lemma~\ref{lem:entropy} corresponds to an average input entropy of $\log_2\mathcal M+\bar H$ bits per channel use. Prescribing the average rather than the entropy of each block allows entropy to be allocated across blocks as the geometry changes. A per-block floor $H_{\rm floor}$, with $0\leq H_{\rm floor}\leq\bar H\leq\log_2R$, limits the concentration of each relative-pattern distribution when positive and setting \(H_{\mathrm{floor}}=0\) disables this additional restriction.
Finally, a probability floor $0<\pi_{\min}<1/R$ ensures that every symbol vector can still be transmitted.

Define $\mathbf V_\tau=[\mathbf v_{\tau,1},\ldots, \mathbf v_{\tau,R}]$ and the collections $\mathcal V=\{\mathbf V_\tau\}_{\tau\in\mathcal T}$ and $\boldsymbol\pi=\{\boldsymbol\pi_\tau\}_{\tau\in\mathcal T}$. Introducing an auxiliary variable $C$ for the worst cell in \eqref{eq:CE}, the proposed PS-SLP design minimizes $C_{\mathrm E}$ in epigraph form as
\begin{equation}
\tag{$\mathbf{P1}$}
\label{prob:main}
\begin{aligned}
\underset{\mathcal V,\boldsymbol\pi,C}
{\operatorname{minimize}}\quad&C\\
\textbf{s.t.}\quad
\mathrm C_1:\;&\eqref{eq:robust_ci},
\quad\forall k,\tau,r,\ b\in\{-1,1\},\\
\mathrm C_2:\;&L\sum_{\tau=1}^{T}M_{\tau,q}^2\leq C,
\quad\forall q\in\mathcal Q,\\
\mathrm C_3:\;&\eqref{eq:peak}\text{--}\eqref{eq:energy},\\
\mathrm C_4:\;&\sum_{r=1}^{R}\pi_{\tau,r}=1,
\quad\forall\tau,\\
&\pi_{\tau,r}\geq\pi_{\min},
\quad\forall\tau,r,\\
\mathrm C_5:\;&L\sum_{\tau=1}^{T}H(\boldsymbol\pi_\tau)
\geq n_{\mathrm{obs}}\bar H,\\
\mathrm C_6:\;&H(\boldsymbol\pi_\tau)\geq H_{\rm floor},
\quad\forall\tau.
\end{aligned}
\end{equation}
Unless otherwise indicated, indices range over $k\in\mathcal K$, $n\in\mathcal N$, $\tau\in\mathcal T$, and $r\in\mathcal R$. Constraint $\mathrm C_1$ enforces robust CI, while $\mathrm C_2$ makes $C$ equal to $C_{\mathrm E}$ at the optimum. Constraint $\mathrm C_3$ collects the power and energy limits, $\mathrm C_4$ defines the admissible pattern distributions, and $\mathrm C_5$ and $\mathrm C_6$ impose the session entropy requirement and the per-block entropy floor.

Let $\mathbf S^{n_{\mathrm{obs}}}$ denote the sequence of all $n_{\mathrm{obs}}$ multiuser symbol vectors. Independence across channel uses and Lemma~\ref{lem:entropy} imply
\begin{equation}
H(\mathbf S^{n_{\mathrm{obs}}})
\hspace{-0.75mm}=\hspace{-0.75mm}n_{\mathrm{obs}}\log_2\mathcal M
+L\sum_{\tau=1}^{T}H(\boldsymbol\pi_\tau)
\hspace{-0.75mm}\geq\hspace{-0.75mm} n_{\mathrm{obs}}(\log_2\mathcal M+\bar H).
\label{eq:Hinput}
\end{equation}
This is an input-entropy requirement, not a coded-throughput guarantee. Minimizing $C$ improves the detection-error guarantee in \eqref{eq:general_covertness}. It does not directly optimize Willie's exact detection error.  The structure of \eqref{prob:main} also reveals how the certificate depends on the session length.

\begin{remark}
\label{rem:length_scaling}
Let $T$, the geometry, the entropy requirements, and the power limits be fixed, with $E_{\mathrm{sess}}\propto n_{\mathrm{obs}}=TL$. Then the constraints of \eqref{prob:main} do not depend on $L$, since $L$ cancels in the energy and session-entropy constraints, whereas the objective is proportional to $L$ through $\mathrm C_2$. Hence, the same design is optimal for every $L$ and needs to be computed only once, and
\begin{equation}
C_{\mathrm E}=c\,n_{\mathrm{obs}},\quad
c=\max_{q\in\mathcal Q}\frac{1}{T}\sum_{\tau=1}^{T}M_{\tau,q}^2,
\end{equation}
where $c$ is evaluated at this design and does not depend on $L$. The same proportionality holds for $C_{\mathrm{sinh}}$. Let $\epsilon_{\min}^{\mathrm{cert}} =\frac12\sqrt{e^{C_{\mathrm E}}-1}$ denote the smallest covertness level certified by $C_{\mathrm E}$ through \eqref{eq:target_covertness}.
Consequently,
\begin{equation}
\epsilon_{\min}^{\mathrm{cert}}(n_{\mathrm{obs}})=\frac{1}{2}\sqrt{e^{c\,n_{\mathrm{obs}}}-1},
\end{equation}
and the largest session length certifiable for a target $\epsilon$ is
\begin{equation}
n_{\mathrm{obs}}^{\max}(\epsilon)=\frac{\ln(1+4\epsilon^2)}{c}.
\label{eq:nmax}
\end{equation}
\end{remark}

\subsection{Alternating Optimization}
\label{subsec:ao}
Problem~\eqref{prob:main} is jointly nonconvex, since the probabilities multiply precoding-dependent terms in the covertness constraint $\mathrm C_2$ and in the average-power and energy constraints. However, it is convex in each block of variables when the other is fixed. We therefore alternate between precoding and probability optimization.

\subsubsection{Precoding update}
At iteration $i$, fix $\boldsymbol\pi=\boldsymbol\pi^{(i)}$
and consider
\begin{equation}
\tag{$\mathbf{P1}$-$\mathbf{V}$}\label{prob:waveform}
\begin{aligned}
\underset{\mathcal V,C}{\operatorname{minimize}}\quad&C\\
\textbf{s.t.}\quad
&\mathrm C_1\text{--}\mathrm C_3
\text{ of }\eqref{prob:main},\\
&\boldsymbol\pi=\boldsymbol\pi^{(i)}.
\end{aligned}
\end{equation}
Constraints $\mathrm C_4$--$\mathrm C_6$ involve only the probabilities and are satisfied by $\boldsymbol\pi^{(i)}$. For fixed probabilities, the remaining difficulty lies in the covertness constraint $\mathrm C_2$. Its left-hand side is convex in the precoding vectors, but it composes a maximum over reference directions, a square, an exponential, and a weighted sum, which conic solvers cannot handle directly. Therefore, we express $\mathrm C_2$ through two standard cones, namely the exponential cone and the rotated second-order cone,
\begin{align}
\mathcal K_{\exp}
&=\{(x,y,z):y>0,\;ye^{x/y}\leq z\}\nonumber\\
&\quad\cup\{(x,0,z):x\leq0,\;z\geq0\},
\label{eq:expcone}\\
\mathcal K_{\mathrm r}
&=\{(x,y,\mathbf z):x,y\geq0,\;2xy\geq\|\mathbf z\|^2\},
\label{eq:rotated_cone}
\end{align}
where the dimension of $\mathbf z$ is determined by the cone argument. To this end, we introduce nonnegative auxiliary variables $a_{\tau,r,q}$, $b_{\tau,r,q}$, $w_{\tau,r,q}$, and $t_{\tau,q}$, which represent the leakage amplitude, the normalized leakage power, the exponential leakage term, and the probability-weighted leakage measure, respectively, and let $\mathbf t_q=[t_{1,q},\ldots,t_{T,q}]^T$.

\begin{proposition}
\label{prop:precoding_convex}
For fixed probabilities $\boldsymbol\pi^{(i)}$, the covertness constraint $\mathrm C_2$ holds if and only if there exist nonnegative $a_{\tau,r,q}$, $b_{\tau,r,q}$, $w_{\tau,r,q}$, and $t_{\tau,q}$ such that
\begin{subequations}
\label{eq:C2_conic}
\begin{align}
&\beta_{\tau,q}^{\max}\bigl|\widetilde{\mathbf a}^{H}(\nu_{\tau,q,i_{\mathrm a}})\mathbf v_{\tau,r}\bigr|
+\bigl(\beta_{\tau,q}^{\max}\delta_{\tau,q}^{a}+\delta_{\tau,q}^{\mathrm{mod}}\bigr)\|\mathbf v_{\tau,r}\|\nonumber\\
&\hspace{30mm}\leq a_{\tau,r,q},\quad\forall\tau,r,q,i_{\mathrm a},\label{eq:C2_conic_a}\\
&\bigl(b_{\tau,r,q},\sigma_{w,\min}^2/2,a_{\tau,r,q}\bigr)\in\mathcal K_{\mathrm r},\quad\forall\tau,r,q,\label{eq:C2_conic_b}\\
&\bigl(b_{\tau,r,q},1,1+w_{\tau,r,q}\bigr)\in\mathcal K_{\exp},\quad\forall\tau,r,q,\label{eq:C2_conic_c}\\
&\textstyle\sum_{r=1}^{R}\pi_{\tau,r}^{(i)}w_{\tau,r,q}\leq t_{\tau,q},\quad\forall\tau,q,\label{eq:C2_conic_d}\\
&\bigl(C,1/2,\sqrt L\,\mathbf t_q\bigr)\in\mathcal K_{\mathrm r},\quad\forall q.\label{eq:C2_conic_e}
\end{align}
\end{subequations}
\end{proposition}
\begin{proof} The proof is presented in Appendix~\ref{app:precoding_cones}.\end{proof}

Consequently, \eqref{prob:waveform} is equivalent to the convex problem
\begin{equation}
\tag{$\mathbf{P2}$-$\mathbf{V}$}
\label{prob:precoding_conic}
\begin{aligned}
\underset{\mathcal V,C,\{a,b,w,t\}}{\operatorname{minimize}}\quad&C\\
\textbf{s.t.}\quad
&\mathrm C_1,\ \mathrm C_3\text{ of }\eqref{prob:main},\quad\eqref{eq:C2_conic},\\
&a_{\tau,r,q},b_{\tau,r,q},w_{\tau,r,q},t_{\tau,q}\geq0,
\end{aligned}
\end{equation}
whose constraints are second-order and exponential cones, since $\mathrm C_1$ combines affine terms with $\|\mathbf v_{\tau,r}\|$ and $\mathrm C_3$ is convex quadratic for fixed probabilities. 

\subsubsection{Probability update}
This optimization updates the block-wise pattern distributions during the design procedure. It is not performed at the symbol rate. During transmission, the data bits of each block can be mapped to pattern sequences with the optimized composition by distribution matching~\cite{Schulte2016}, which requires no online optimization.
For the updated precoding vectors $\mathcal V^{(i+1)}$,
define
\begin{equation}
f_{\tau,r,q}^{(i+1)}\triangleq
\exp\!\left(
\rho_{\tau,r,q}(\mathbf v_{\tau,r}^{(i+1)})
\right)-1.
\label{eq:fixed_f}
\end{equation}
The probability update considers
\begin{equation}
\tag{$\mathbf{P1}$-$\mathbf{P}$}\label{prob:probability}
\begin{aligned}
\underset{\boldsymbol\pi,C}{\operatorname{minimize}}\quad&C\\
\textbf{s.t.}\quad
&L\sum_{\tau=1}^{T}
\left(\sum_{r=1}^{R}
\pi_{\tau,r}f_{\tau,r,q}^{(i+1)}\right)^2\\
&\qquad\leq C,\quad\forall q\in\mathcal Q,\\
&\eqref{eq:avgfeed},\quad\eqref{eq:energy},\\
&\mathrm C_4\text{--}\mathrm C_6
\text{ of }\eqref{prob:main},\\
&\mathcal V=\mathcal V^{(i+1)}.
\end{aligned}
\end{equation}
For fixed precoding vectors, $M_{\tau,q}=\sum_r\pi_{\tau,r}f_{\tau,r,q}^{(i+1)}$ is linear in the probabilities, so the covertness constraints are convex quadratics. Let $\mathbf m_q^{(i+1)}(\boldsymbol\pi)$ denote the $T$-dimensional vector with entries $M_{\tau,q}$. The entropy constraints $\mathrm C_5$ and $\mathrm C_6$ are convex but, like $\mathrm C_2$, not directly in conic form. We therefore introduce real auxiliary variables $h_{\tau,r}$, which represent the entropy terms $\pi_{\tau,r}\ln\pi_{\tau,r}$.

\begin{proposition}
\label{prop:probability_convex}
Constraints $\mathrm C_5$ and $\mathrm C_6$ hold if and only if there exist real $h_{\tau,r}$ such that
\begin{subequations}
\label{eq:entropy_conic}
\begin{align}
&(-h_{\tau,r},\pi_{\tau,r},1)\in\mathcal K_{\exp},\quad\forall\tau,r,\label{eq:entropy_conic_a}\\
&L\textstyle\sum_{\tau=1}^{T}\sum_{r=1}^{R}h_{\tau,r}\leq-n_{\mathrm{obs}}\bar H\ln2,\label{eq:entropy_conic_b}\\
&\textstyle\sum_{r=1}^{R}h_{\tau,r}\leq-H_{\mathrm{floor}}\ln2,\quad\forall\tau.\label{eq:entropy_conic_c}
\end{align}
\end{subequations}
\end{proposition}
\begin{proof} The proof is presented in Appendix~\ref{app:probability_cones}.\end{proof}

Consequently, \eqref{prob:probability} is equivalent to the convex problem
\begin{equation}
\tag{$\mathbf{P2}$-$\mathbf{P}$}
\label{prob:probability_conic}
\begin{aligned}
\underset{\boldsymbol\pi,C,\{h\}}{\operatorname{minimize}}\quad&C\\
\textbf{s.t.}\quad
&\bigl(C,1/2,\sqrt L\,\mathbf m_q^{(i+1)}(\boldsymbol\pi)\bigr)\in\mathcal K_{\mathrm r},\quad\forall q,\\
&\eqref{eq:avgfeed},\ \eqref{eq:energy},\ \mathrm C_4\text{ of }\eqref{prob:main},\ \eqref{eq:entropy_conic},
\end{aligned}
\end{equation}
whose constraints are affine, second-order, and exponential cones. Both updates admit exact convex reformulations and are solved numerically to the prescribed tolerance, without successive approximation of the subproblems.
Beyond its tractability, the probability update also reveals the structure of the optimal shaping.

\begin{remark}
\label{rem:gibbs}
For fixed precoding vectors, suppress the iteration superscript in \eqref{eq:fixed_f}. Since \eqref{prob:probability} is convex, Karush--Kuhn--Tucker (KKT) multipliers exist at its optimum under Slater's condition, which holds, e.g., if $\bar H<\log_2R$ and the uniform distribution satisfies the average-power and energy constraints. At such an optimum, define
\begin{equation}
\begin{aligned}
\psi_{\tau,r}
={}&2L\sum_{q\in\mathcal Q}\lambda_{C,q}M_{\tau,q}f_{\tau,r,q}\\
&+\sum_{n\in\mathcal N}\lambda_{n,\tau}|v_{n,\tau,r}|^2+\lambda_E LT_s\|\mathbf v_{\tau,r}\|^2,
\end{aligned}
\label{eq:psi}
\end{equation}
where $\lambda_{C,q}$, $\lambda_{n,\tau}$, and $\lambda_E$ are the nonnegative multipliers of the cell-covertness, average-power, and session-energy constraints, respectively, and $M_{\tau,q}=\sum_r\pi_{\tau,r}f_{\tau,r,q}$. Let $\lambda_H$ and $\lambda_{H,\tau}$ be the nonnegative multipliers of $\mathrm C_5$ and $\mathrm C_6$, written in $\leq0$ form, and set $\vartheta_\tau=(L\lambda_H+\lambda_{H,\tau})/\ln2$. For $\vartheta_\tau>0$, KKT stationarity gives
\begin{equation}
\pi_{\tau,r}\propto\exp\!\left(-\frac{\psi_{\tau,r}}{\vartheta_\tau}\right),\quad\pi_{\tau,r}>\pi_{\min}.
\label{eq:gibbs}
\end{equation}
Thus, patterns above the probability floor follow an exponential law in their marginal leakage-and-power cost, analogous to Maxwell--Boltzmann shaping \cite{Bocherer2015}. The parameter $\vartheta_\tau$, set by the entropy constraints, plays the role of a temperature. A stricter entropy requirement flattens the distribution, which becomes uniform for $\bar H=\log_2R$. If the per-block floor $\mathrm C_6$ is inactive, then $\lambda_{H,\tau}=0$ and $\vartheta_\tau=L\lambda_H/\ln2$ is common to all blocks. The blocks then differ only in their costs, and the session constraint $\mathrm C_5$ allocates entropy according to the block geometry.
\end{remark}

\subsection{Initialization, Convergence, and Complexity}
\label{subsec:convergence}
The alternating optimization requires a feasible starting point. The uniform distribution $\pi_{\tau,r}=1/R$ satisfies the probability and entropy constraints and is used as a first initialization for \eqref{prob:waveform}. Since the average-power and energy constraints depend on the probabilities, \eqref{prob:waveform} may be infeasible for this choice, in which case a few other admissible distributions are tried. Algorithm~\ref{alg:ao} summarizes the procedure once a feasible pair has been obtained.
With each update solved exactly and the previous iterate remaining feasible for the next subproblem, the objective of Algorithm~\ref{alg:ao} is nonincreasing. Being bounded below by zero through $\mathrm C_2$, it converges to a finite limit.

\begin{figure}[t]
\centering
\begin{minipage}{\columnwidth}
\hrule\smallskip
\refstepcounter{algorithm}\textbf{Algorithm \thealgorithm: Alternating PS-SLP Design}\label{alg:ao}
\smallskip\hrule\smallskip
\textbf{Input:} Channels, cell envelopes, CI and power requirements,
$\bar H$, $H_{\rm floor}$, $\pi_{\min}$, tolerance $\eta>0$,
and iteration limit $I_{\max}$.\\
\textbf{Initialize:} A feasible $(\mathcal V^{(0)},\boldsymbol\pi^{(0)})$, with $C^{(0)}=C_{\mathrm E}(\mathcal V^{(0)},\boldsymbol\pi^{(0)})$ and $i=0$.
\begin{enumerate}\setlength{\itemsep}{1pt}\setlength{\parskip}{0pt}
\item Solve \eqref{prob:precoding_conic} for $\mathcal V^{(i+1)}$.
\item Form \eqref{eq:fixed_f} and solve \eqref{prob:probability_conic} for $\boldsymbol\pi^{(i+1)}$.
\item Evaluate $C^{(i+1)}=C_{\mathrm E}(\mathcal V^{(i+1)},\boldsymbol\pi^{(i+1)})$ and check all constraints in \eqref{prob:main}. Accept the update only if all constraints hold and the objective does not increase, within the solver tolerances. Otherwise, return the last accepted pair.

\item Stop if
$|C^{(i+1)}-C^{(i)}|/\max\{C^{(i)},10^{-12}\}\leq\eta$
or $i+1=I_{\max}$. Otherwise increment $i$ and repeat.
\end{enumerate}
\textbf{Output:} The last accepted pair, $C_{\mathrm E}$, $C_{\mathrm{sinh}}$, and the corresponding guarantees from \eqref{eq:general_covertness}.
\smallskip\hrule
\end{minipage}
\end{figure}

The partition of $\mathcal A_W$ contains many cells, and including all of them in every conic solve would be costly. We therefore use constraint generation: each solve includes only a subset of $Q_{\mathrm a}$ active cells, the solution is then checked over all $Q$ cells, and any violating cells are added before solving again. Since the reduced problem is a relaxation of the full one, a solution that satisfies all cells is also optimal for the full subproblem. For $Q_{\mathrm a}$ active cells and $I_{\mathrm a}$ reference directions per cell, the precoding update has $2TRN$ real variables, $2KTR$ robust CI inequalities, $NTR$ peak-power, $NT$ average-power, and one energy constraint. Its conic representation in Proposition~\ref{prop:precoding_convex} adds $O(TRQ_{\mathrm a})$ auxiliary variables and exponential cones and $O(I_{\mathrm a}TRQ_{\mathrm a})$ second-order cones for the leakage envelopes. The probability update has $TR$ probabilities, $TR$ entropy-lifting variables, and $Q_{\mathrm a}$ quadratic covertness constraints.

Let $T_V(Q_{\mathrm a})$ and $T_\pi(Q_{\mathrm a})$ denote the costs of the two conic solves, which depend on sparsity and solver accuracy. With $I$ alternating iterations and at most $G$ constraint-generation rounds per update, the cost per initialization is
\begin{equation}
O\!\left(IG\left[T_V(Q_{\mathrm a})+T_\pi(Q_{\mathrm a})+I_{\mathrm a}NTRQ\right]\right),
\end{equation}
where the last term accounts for checking the certificates over all cells. The main scaling limitations are $R=\mathcal M^{K-1}$ and the $TRQ$ spatial-envelope terms. With $S$ initializations, the total design cost is $S$ times the above cost. The resulting SLP tables require $O(TRN)$ complex storage. During transmission, selecting and rotating a stored vector costs $O(N)$ operations per channel use, excluding pattern generation, giving $O(n_{\mathrm{obs}}N)$ online cost per session. No conic optimization is performed at the symbol rate.

\subsection{Proposed Designs and Benchmarks}
\label{subsec:benchmarks}
The proposed PS-SLP and PS-BLP designs are compared with their uniform-probability benchmarks, uniform symbol-level precoding (U-SLP) and uniform block-level precoding (U-BLP), under the same channel, CI, power, and energy requirements. PS-SLP jointly optimizes all precoding vectors and the pattern probabilities, as in \eqref{prob:main}. PS-BLP jointly optimizes the probabilities and a common linear precoder per transmission block, restricting
\begin{equation}
\mathbf v_{\tau,r}=\mathbf W_\tau\mathbf u_r,\quad\mathbf W_\tau\in\mathbb C^{N\times K}.
\label{eq:blp}
\end{equation}
Since $\mathbf s_{\tau,\ell}=\zeta_{\tau,\ell}\mathbf u_{\mathsf r_{\tau,\ell}}$, the transmitted signal becomes $\mathbf x_{\tau,\ell}=\mathbf W_\tau\mathbf s_{\tau,\ell}$, i.e., block-level linear precoding designed under the same robust pattern-wise CI requirements. Thus, PS-BLP retains interference exploitation, consistent with CI-based block-level precoding~\cite{Wang2025CIBLP}. Its distinction from PS-SLP is the common linear mapping across all relative patterns. This reduces the precoding variables to $TNK$ complex coefficients while retaining all relative-pattern constraints. Since \eqref{eq:blp} is linear in $\mathbf W_\tau$, the precoding update remains convex and Algorithm~\ref{alg:ao} applies unchanged. U-SLP and U-BLP fix $\pi_{\tau,r}=1/R$ and optimize only the precoding vectors of the respective designs. Comparing each proposed design with its uniform counterpart isolates the benefit of shaping, whereas comparing U-SLP with U-BLP assesses the benefit of independent pattern-wise precoding over a common CI-constrained linear precoder.

Finally, to isolate the benefit of allocating entropy across blocks, we consider a per-block entropy benchmark, which replaces $\mathrm C_5$ and $\mathrm C_6$ by $H(\boldsymbol\pi_\tau)\geq\bar H$ for every block. Each block must then meet the entropy requirement on its own, so entropy cannot be shifted from blocks with unfavorable geometry to blocks with favorable geometry. Since its feasible set is contained in that of \eqref{prob:main} when $H_{\rm floor}\leq\bar H$, it cannot achieve a lower optimal $C_{\mathrm E}$.

\section{Numerical Results}
\label{Num}

This section evaluates the proposed designs and certificates. We first describe the simulation setup, then compare the certificates with the exact mixture, and finally examine the effects of the communication requirements, the uncertainty, the input entropy, the satellite elevation, the power budget, and the session length.

We use $\mathcal M=4$ (QPSK), which is employed in the low-SNR modes of DVB-S2~\cite{DVBS2}.
Unless otherwise stated, the simulation parameters are $N=24$, $K=3$, $T=3$, and $L=10$, giving $R=16$ relative patterns and $n_{\mathrm{obs}}=30$ active channel uses. The satellite altitude is $600$~km and the carrier frequency is $20$~GHz, consistent with the LEO Ka-band scenario in \cite{3GPP38811}. The bandwidth is $20$~MHz, with $T_s=50$~ns. The satellite element gain is $10\log G_s=6$~dBi, while the receive-antenna gains are $10\log G_k=10\log G_W=36.7$~dBi. The additional propagation loss is $\Lambda_{u,\tau}=2$~dB for every link and transmission block. 
The receiver-noise variances are $\sigma_k^2=\sigma_w^2=6.69\times10^{-14}$~W, and we set $\sigma_{w,\min}^2=\sigma_w^2$.

The satellite follows a northbound circular pass along the meridian through the ground-coordinate origin, with the array axis $\mathbf o_{a,\tau}$ parallel to its instantaneous velocity. Its position at the session midpoint is specified by the elevation angle $\theta_{\mathrm{el}}^{\mathrm{mid}}$, measured from the local horizon at the origin, with $90^\circ$ corresponding to the satellite directly overhead. The nominal value is $\theta_{\mathrm{el}}^{\mathrm{mid}}=60^\circ$, with the satellite north of the origin, and the block centers are separated by $2$~s. 
Thus, each block lasts $LT_s=0.5~\mu$s, so the session consists of short bursts distributed over the pass. Such short transmissions are inherent in covert signaling, since the square-root law limits the number of channel uses that can remain covert \cite{Bash2013}, as also quantified for the proposed certificates in Remark~\ref{rem:length_scaling} and Fig.~\ref{fig:satellite_5}. 
The Bobs' local east--north coordinates, in kilometers, are $(-30,-105)$, $(20,-25)$, and $(-10,105)$. Willie's uncertainty region is centered at $(70,35)$ and has dimensions $50\alpha_W\times40\alpha_W$~km$^2$. It is partitioned into $1\times1$~km cells on a grid aligned with the region center, with boundary cells extended to full cells, giving $Q=520$, $2000$, and $4560$ for $\alpha_W=0.5$, $1$, and $1.5$, respectively. The leakage envelope of each cell uses $I_{\mathrm a}=4$ reference directions.

The nominal design parameters are $\gamma_{\mathrm{req}}=4$~dB, $\kappa_B=2\%$, $\bar H/\log_2R=0.75$, $H_{\mathrm{floor}}=0$, $\pi_{\min}=10^{-3}$, and $\alpha_W=1$. All RF chains share the same limits, $P_n^{\rm pk}=P^{\mathrm{pk}}=0.5$~W and $P_{n,\tau}^{\rm av}=P^{\mathrm{av}}=P^{\mathrm{pk}}/2$. The session energy budget is $E_{\mathrm{sess}}=\beta_E Nn_{\mathrm{obs}}T_sP^{\mathrm{av}}$, with $\beta_E=1$ unless otherwise stated. For $\beta_E=1$, the energy constraint is implied by the per-RF-chain average-power constraints, whereas $\beta_E<1$ imposes an additional session-wide limitation.

The optimization is implemented in MATLAB using CVX and MOSEK, with three probability initializations. Across each parameter sweep, previously obtained designs are also used to initialize subsequent operating points whenever they remain feasible under the new parameter values. 
Except for the certificate comparison in Fig.~\ref{fig:satellite_4}, the figures report $\epsilon_{\min}^{\mathrm{cert}}$, defined in Remark~\ref{rem:length_scaling}, which is computed from $C_{\mathrm E}$ at the retained feasible solution. A smaller value gives a stronger certified detection-error guarantee. Values above $1$ are omitted, and the horizontal line at $\epsilon=1$ marks the limit beyond which the certificate is vacuous. The tighter certificate $C_{\mathrm{sinh}}$ provides an additional post-design guarantee, illustrated in Fig.~\ref{fig:satellite_4}.

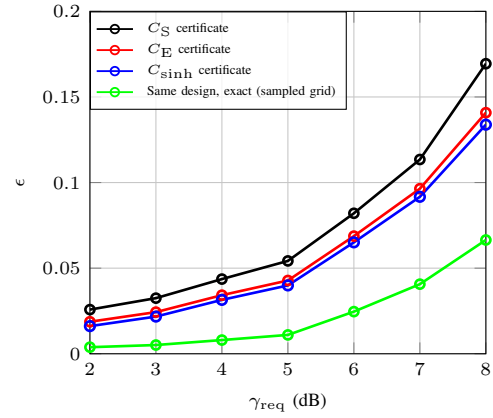
\begin{figure}[t]
\centering
\begin{tikzpicture}
\begin{axis}[
    width=0.77\linewidth,
    height=0.69\linewidth,
    ylabel={$\epsilon$},
    ymin=0,
    scaled x ticks=false,
    scaled y ticks=false,
    x tick label style={
        /pgf/number format/fixed,
        /pgf/number format/precision=0,
        font=\scriptsize
    },
    y tick label style={
        /pgf/number format/fixed,
        /pgf/number format/precision=2,
        font=\scriptsize
    },
    grid=both,
    minor grid style={gray!20},
    major grid style={gray!45},
    label style={font=\scriptsize},
    unbounded coords=jump,
    legend columns=1,
    legend cell align=left,
    legend style={
        font=\tiny,
        at={(0.0,1)},
        anchor=north west,
        draw=black,
        fill=white,
        fill opacity=0.8,
        text opacity=1,
        inner sep=2pt,
        row sep=-1pt,
        column sep=3pt
    },
    legend image post style={line width=0.7pt,scale=0.8},
    xlabel={$\gamma_{\mathrm{req}}$ (dB)},
    xmin=2,xmax=8,xtick={2,3,4,5,6,7,8},
    ymax=0.2,ytick={0,0.05,0.10,0.15,0.20,0.25,0.30,0.35,0.40}
]
\addplot[
    black, mark=o,
    mark options={solid,fill=white},
    mark size=1.7pt,
    line width=1pt,
    solid,
    forget plot
]
table[x=GammaDb,y=EpsilonCS,col sep=space]
{Figures4/fig4_data.dat};

\addplot[
    red, mark=o,
    mark options={solid,fill=white},
    mark size=1.7pt,
    line width=1pt,
    solid,
    forget plot
]
table[x=GammaDb,y=EpsilonCE,col sep=space]
{Figures4/fig4_data.dat};

\addplot[
    green, mark=o,
    mark options={solid,fill=white},
    mark size=1.7pt,
    line width=1pt,
    solid,
    forget plot
]
table[x=GammaDb,y=EpsilonGrid,col sep=space]
{Figures4/fig4_data.dat};

\addplot[
    blue,
    mark=o,
    mark options={solid,fill=white},
    mark size=1.8pt,
    line width=1pt,
    solid,
    forget plot
]
table[x=GammaDb,y=EpsilonCsinh,col sep=space]
{Figures4/fig4_data.dat};

\addlegendimage{black,solid,mark=o,mark options={solid,fill=white}}
\addlegendentry{$C_{\mathrm S}$ certificate}
\addlegendimage{red,solid,mark=o,mark options={solid,fill=white}}
\addlegendentry{$C_{\mathrm E}$ certificate}
\addlegendimage{blue,solid,mark=o,
mark options={solid,fill=white}}
\addlegendentry{$C_{\mathrm{sinh}}$ certificate}
\addlegendimage{green,solid,mark=o,mark options={solid,fill=white}}
\addlegendentry{Same design, exact (sampled grid)}
\end{axis}
\end{tikzpicture}
\caption{Covertness parameter ($\epsilon$) versus the CI QoS threshold for different certificates and the sampled exact-mixture diagnostic, evaluated for the same PS-SLP designs.}
\label{fig:satellite_4}
\end{figure}

Fig.~\ref{fig:satellite_4} compares the certificates of Section~\ref{sec:covertness} for the same PS-SLP designs. The proposed certificate $C_{\mathrm E}$ is $17$--$28\%$ tighter than $C_{\mathrm S}$ over the entire range, in line with the low-leakage analysis, since it avoids the variance penalty of the simpler bound. Evaluating $C_{\mathrm{sinh}}$ for the same designs tightens the certificate by a further $5$--$14\%$, with the largest improvement at low QoS thresholds. To assess how conservative the certificates are, the figure also evaluates the same designs, without any reoptimization, through the exact divergence of Proposition~\ref{prop:exact}, maximized over an $81\times65$ grid of Willie locations. At these locations, the resulting $\epsilon$ is $2.1$ to $4.9$ times lower than the one certified by $C_{\mathrm E}$. Hence, the actual covertness of the proposed designs is considerably better than what the certificate guarantees. This margin is the price of guaranteeing covertness at every position of the continuous region rather than at sampled points, and it narrows as the QoS threshold increases. All curves increase with $\gamma_{\mathrm{req}}$, since stricter CI requirements leave less freedom for suppressing leakage toward Willie.

\begin{figure}[t]
\centering
\begin{tikzpicture}
\begin{axis}[
    width=0.77\linewidth,
    height=0.69\linewidth,
    ylabel={$\epsilon_{\min}^{\mathrm{cert}}$},
    ymin=0,
    scaled x ticks=false,
    scaled y ticks=false,
    x tick label style={
        /pgf/number format/fixed,
        /pgf/number format/precision=0,
        font=\scriptsize
    },
    y tick label style={
        /pgf/number format/fixed,
        /pgf/number format/precision=1,
        font=\scriptsize
    },
    grid=both,
    minor grid style={gray!20},
    major grid style={gray!45},
    label style={font=\scriptsize},
    unbounded coords=jump,
    legend columns=2,
    legend cell align=left,
    legend style={
        font=\tiny,
        at={(0.0,1)},
        anchor=north west,
        draw=black,
        fill=white,
        fill opacity=0.8,
        text opacity=1,
        inner sep=2pt,
        row sep=-1pt,
        column sep=3pt
    },
    legend image post style={line width=0.7pt,scale=0.8},
 xlabel={$\gamma_{\mathrm{req}}$ (dB)},
 xmin=2,xmax=8,xtick={2,3,4,5,6,7,8},ymax=1.05,ytick={0,0.2,0.4,0.6,0.8,1}
]
\addplot[gray!65,densely dotted,line width=0.8pt,forget plot] coordinates {(2,1) (8,1)};

\addplot[
    blue, mark=o,
    mark options={solid,fill=white},
    mark size=1.7pt,
    line width=1pt,
    solid,
    forget plot
]
table[x=GammaReq_PS_SLP_k0,y=PS_SLP_k0,col sep=space]
{Figures4/fig1_data.dat};

\addplot[
    blue, mark=square,
    mark options={solid,fill=white},
    mark size=1.7pt,
    line width=1pt,
    dashed,
    forget plot
]
table[x=GammaReq_U_SLP_k0,y=U_SLP_k0,col sep=space]
{Figures4/fig1_data.dat};

\addplot[
    red, mark=o,
    mark options={solid,fill=white},
    mark size=1.7pt,
    line width=1pt,
    solid,
    forget plot
]
table[x=GammaReq_PS_SLP_k2,y=PS_SLP_k2,col sep=space]
{Figures4/fig1_data.dat};

\addplot[
    red, mark=square,
    mark options={solid,fill=white},
    mark size=1.7pt,
    line width=1pt,
    dashed,
    forget plot
]
table[x=GammaReq_U_SLP_k2,y=U_SLP_k2,col sep=space]
{Figures4/fig1_data.dat};

\addplot[
    red, mark=triangle,
    mark options={solid,fill=white},
    mark size=1.7pt,
    line width=1pt,
    dotted,
    forget plot
]
table[x=GammaReq_PS_BLP_k2,y=PS_BLP_k2,col sep=space]
{Figures4/fig1_data.dat};

\addplot[
    red, mark=diamond,
    mark options={solid,fill=white},
    mark size=1.7pt,
    line width=1pt,
    dashdotted,
    mark indices={1,2,3,4,5,6,7},
    forget plot
]
table[x=GammaReq_U_BLP_k2,y=U_BLP_k2,col sep=space]
{Figures4/fig1_data.dat};

\addplot[
    black, mark=o,
    mark options={solid,fill=white},
    mark size=1.25pt,
    line width=1pt,
    solid,
    forget plot
]
table[x=GammaReq_PS_SLP_k5,y=PS_SLP_k5,col sep=space]
{Figures4/fig1_data.dat};

\addplot[
    black, mark=square,
    mark options={solid,fill=white},
    mark size=1.7pt,
    line width=1pt,
    dashed,
    mark indices={1,2,3,4,5,6},
    forget plot
]
table[x=GammaReq_U_SLP_k5,y=U_SLP_k5,col sep=space]
{Figures4/fig1_data.dat};

\addplot[black,only marks,mark=square*,mark size=2pt,forget plot]
coordinates {(7.105171999035,1)};

\addlegendimage{black,solid,mark=o,mark options={solid,fill=white}}
\addlegendentry{PS-SLP}
\addlegendimage{black,dashed,mark=square,mark options={solid,fill=white}}
\addlegendentry{U-SLP}
\addlegendimage{black,dotted,mark=triangle,mark options={solid,fill=white}}
\addlegendentry{PS-BLP}
\addlegendimage{black,dashdotted,mark=diamond,mark options={solid,fill=white}}
\addlegendentry{U-BLP}
\addlegendimage{blue,solid,mark=none,line width=0.9pt}
\addlegendentry{$\kappa_B=0\%$}
\addlegendimage{red,solid,mark=none,line width=0.9pt}
\addlegendentry{$\kappa_B=2\%$}
\addlegendimage{black,solid,mark=none,line width=0.9pt}
\addlegendentry{$\kappa_B=5\%$}
\end{axis}
\end{tikzpicture}
\caption{Minimum certified covertness level ($\epsilon_{\min}^{\mathrm{cert}}$) versus the CI QoS threshold under Bob CSI uncertainty.}
\label{fig:satellite_1}
\end{figure}
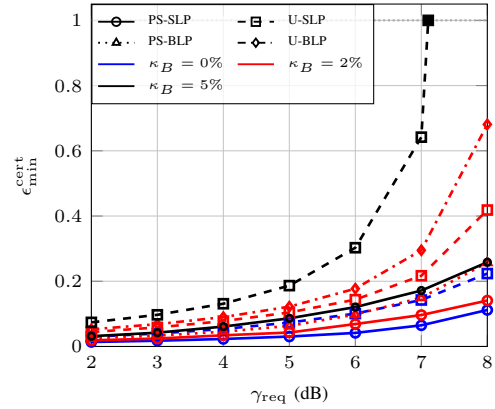

Fig.~\ref{fig:satellite_1} shows $\epsilon_{\min}^{\mathrm{cert}}$ versus the CI QoS threshold for $\kappa_B\in\{0\%,2\%,5\%\}$. Higher thresholds require larger received margins at the Bobs, which leaves less freedom to suppress leakage toward Willie. Bob-channel uncertainty has the same effect, since the CI requirements must hold for every channel error. PS-SLP achieves the lowest certified level in all cases. At the nominal point ($\gamma_{\mathrm{req}}=4$~dB, $\kappa_B=2\%$), it certifies $\epsilon=0.034$, compared with $0.078$ for U-SLP, $0.045$ for PS-BLP, and $0.090$ for U-BLP. Comparing each shaped design with its uniform counterpart shows that shaping reduces the certified level by about half for both precoding structures. Notably, PS-BLP outperforms U-SLP over the whole range, i.e., in this setting, shaping alone yields a larger gain than symbol-level precoding alone. Shaping also extends the range of thresholds that can be certified. For $\kappa_B=5\%$, U-SLP reaches the vacuous level $\epsilon=1$ above $7$~dB, whereas PS-SLP still certifies $\epsilon\approx0.26$ at $8$~dB.

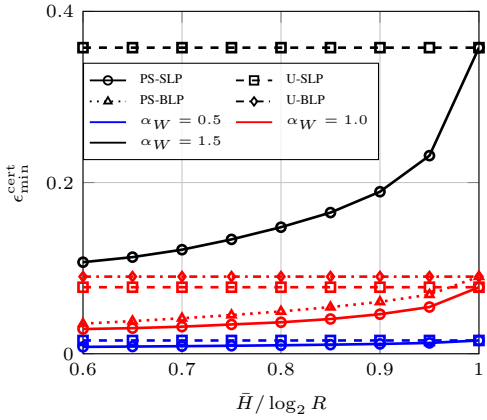
\begin{figure}[t]
\centering
\begin{tikzpicture}
\begin{axis}[
    width=0.77\linewidth,
    height=0.69\linewidth,
    ylabel={$\epsilon_{\min}^{\mathrm{cert}}$},
    ymin=0,
    scaled x ticks=false,
    scaled y ticks=false,
    x tick label style={
        /pgf/number format/fixed,
        /pgf/number format/precision=2,
        font=\scriptsize
    },
    y tick label style={
        /pgf/number format/fixed,
        /pgf/number format/precision=1,
        font=\scriptsize
    },
    grid=both,
    minor grid style={gray!20},
    major grid style={gray!45},
    label style={font=\scriptsize},
    unbounded coords=jump,
    legend columns=2,
    legend cell align=left,
    legend style={
        font=\tiny,
        at={(0.0,0.85)},
        anchor=north west,
        draw=black,
        fill=white,
        fill opacity=0.8,
        text opacity=1,
        inner sep=2pt,
        row sep=-1pt,
        column sep=3pt
    },
    legend image post style={line width=0.7pt,scale=0.8},
 xlabel={$\bar H/\log_2 R$},
 xmin=0.6,xmax=1,xtick={0.6,0.7,0.8,0.9,1.0},ymax=0.4,ytick={0,0.2,0.4,0.6,0.8,1}
]
\addplot[gray!65,densely dotted,line width=0.8pt,forget plot] coordinates {(0.6,1) (1,1)};
\addplot[
    blue, mark=o,
    mark options={solid,fill=white},
    mark size=1.7pt,
    line width=1pt,
    solid,
    forget plot
]
table[x=eta_PS_SLP_a05,y=PS_SLP_a05,col sep=space]
{Figures4/fig2_data.dat};

\addplot[
    blue, mark=square,
    mark options={solid,fill=white},
    mark size=1.7pt,
    line width=1pt,
    dashed,
    forget plot
]
table[x=eta_U_SLP_a05,y=U_SLP_a05,col sep=space]
{Figures4/fig2_data.dat};

\addplot[
    red, mark=o,
    mark options={solid,fill=white},
    mark size=1.7pt,
    line width=1pt,
    solid,
    forget plot
]
table[x=eta_PS_SLP_a10,y=PS_SLP_a10,col sep=space]
{Figures4/fig2_data.dat};

\addplot[
    red, mark=square,
    mark options={solid,fill=white},
    mark size=1.7pt,
    line width=1pt,
    dashed,
    forget plot
]
table[x=eta_U_SLP_a10,y=U_SLP_a10,col sep=space]
{Figures4/fig2_data.dat};

\addplot[
    red, mark=triangle,
    mark options={solid,fill=white},
    mark size=1.7pt,
    line width=1pt,
    dotted,
    forget plot
]
table[x=eta_PS_BLP_a10,y=PS_BLP_a10,col sep=space]
{Figures4/fig2_data.dat};

\addplot[
    red, mark=diamond,
    mark options={solid,fill=white},
    mark size=1.7pt,
    line width=1pt,
    dashdotted,
    forget plot
]
table[x=eta_U_BLP_a10,y=U_BLP_a10,col sep=space]
{Figures4/fig2_data.dat};

\addplot[
    black, mark=o,
    mark options={solid,fill=white},
    mark size=1.7pt,
    line width=1pt,
    solid,
    forget plot
]
table[x=eta_PS_SLP_a15,y=PS_SLP_a15,col sep=space]
{Figures4/fig2_data.dat};

\addplot[
    black, mark=square,
    mark options={solid,fill=white},
    mark size=1.7pt,
    line width=1pt,
    dashed,
    forget plot
]
table[x=eta_U_SLP_a15,y=U_SLP_a15,col sep=space]
{Figures4/fig2_data.dat};

\addlegendimage{black,solid,mark=o,mark options={solid,fill=white}}
\addlegendentry{PS-SLP}
\addlegendimage{black,dashed,mark=square,mark options={solid,fill=white}}
\addlegendentry{U-SLP}
\addlegendimage{black,dotted,mark=triangle,mark options={solid,fill=white}}
\addlegendentry{PS-BLP}
\addlegendimage{black,dashdotted,mark=diamond,mark options={solid,fill=white}}
\addlegendentry{U-BLP}
\addlegendimage{blue,solid,mark=none,line width=0.9pt}
\addlegendentry{$\alpha_W=0.5$}
\addlegendimage{red,solid,mark=none,line width=0.9pt}
\addlegendentry{$\alpha_W=1.0$}
\addlegendimage{black,solid,mark=none,line width=0.9pt}
\addlegendentry{$\alpha_W=1.5$}
\end{axis}
\end{tikzpicture}
\caption{Minimum certified covertness level ($\epsilon_{\min}^{\mathrm{cert}}$) versus the normalized entropy requirement under Willie-location uncertainty.}
\label{fig:satellite_2}
\end{figure}


Fig.~\ref{fig:satellite_2} shows the effect of the normalized entropy requirement for $\alpha_W\in\{0.5,1,1.5\}$, i.e., Willie uncertainty regions of $25\times20$, $50\times40$, and $75\times60$~km$^2$. A larger $\bar H$ forces the probabilities closer to uniform and leaves less freedom to favor low-leakage patterns, so the certified level of the shaped designs increases. At $\bar H/\log_2R=1$, only the uniform distribution is feasible, and each shaped design coincides with its uniform counterpart, whose certified level does not depend on $\bar H$. A larger uncertainty region raises the certified level, since the precoding vectors must suppress leakage over more locations. However, the benefit of shaping grows with the region. At the nominal entropy, PS-SLP reduces the certified level relative to U-SLP by $40\%$ for $\alpha_W=0.5$, $56\%$ for $\alpha_W=1$, and $63\%$ for $\alpha_W=1.5$. When spatial suppression alone becomes harder, making the high-leakage patterns rare becomes increasingly valuable.

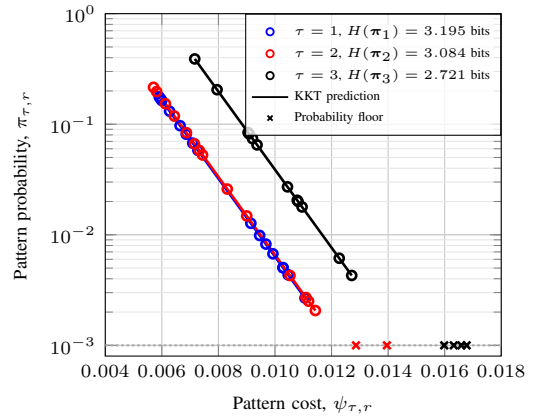
\begin{figure}[t]
\centering
\pgfplotstableread[col sep=space]{Figures4/remark2_data.dat}\remarkData
\pgfplotstablesort[sort key=psi1]\remarkOne{\remarkData}
\pgfplotstablesort[sort key=psi2]\remarkTwo{\remarkData}
\pgfplotstablesort[sort key=psi3]\remarkThree{\remarkData}
\begin{tikzpicture}
\begin{semilogyaxis}[
    width=0.77\linewidth,
    height=0.69\linewidth,
    xlabel={Pattern cost, $\psi_{\tau,r}$},
    ylabel={Pattern probability, $\pi_{\tau,r}$},
    xmin=0.004,xmax=0.018,
    ymin=0.0008,ymax=1,
    xtick={0.004,0.006,0.008,0.010,0.012,0.014,0.016,0.018},
    ytick={0.001,0.01,0.1,1},
    scaled x ticks=false,
    x tick label style={
        /pgf/number format/fixed,
        /pgf/number format/precision=3,
        /pgf/number format/fixed zerofill,
        font=\scriptsize
    },
    y tick label style={font=\scriptsize},
    label style={font=\scriptsize},
    grid=both,
    minor grid style={gray!20},
    major grid style={gray!45},
    unbounded coords=jump,
    every axis plot/.append style={line width=1pt,mark size=1.7pt},
    legend columns=1,
    legend cell align=left,
    legend style={
        font=\tiny,
        at={(1,1)},anchor=north east,
        draw=black,fill=white,fill opacity=0.8,text opacity=1,
        inner sep=2pt,row sep=-1pt
    },
    legend image post style={line width=0.7pt,scale=0.8}
]
\addplot[gray!65,densely dotted,line width=0.8pt,forget plot]
    coordinates {(0.004,0.001) (0.018,0.001)};

\addplot[blue,solid,mark=none,forget plot]
    table[x=psi1,y expr={\thisrow{piKKT1}<0.001 ? nan : \thisrow{piKKT1}}]{\remarkOne};
\addplot[red,solid,mark=none,forget plot]
    table[x=psi2,y expr={\thisrow{piKKT2}<0.001 ? nan : \thisrow{piKKT2}}]{\remarkTwo};
\addplot[black,solid,mark=none,forget plot]
    table[x=psi3,y expr={\thisrow{piKKT3}<0.001 ? nan : \thisrow{piKKT3}}]{\remarkThree};

\addplot[blue,only marks,mark=o,mark options={solid,fill=white},forget plot]
    table[x=psi1,y=piFree1]{\remarkData};
\addplot[red,only marks,mark=o,mark options={solid,fill=white},forget plot]
    table[x=psi2,y=piFree2]{\remarkData};
\addplot[black,only marks,mark=o,mark options={solid,fill=white},forget plot]
    table[x=psi3,y=piFree3]{\remarkData};

\addplot[red,only marks,mark=x,mark size=2pt,forget plot]
    table[x=psi2,y=piFloor2]{\remarkData};
\addplot[black,only marks,mark=x,mark size=2pt,forget plot]
    table[x=psi3,y=piFloor3]{\remarkData};

\addlegendimage{blue,only marks,mark=o,mark options={solid,fill=white}}
\addlegendentry{$\tau=1$, $H(\boldsymbol\pi_1)=3.195$ bits}
\addlegendimage{red,only marks,mark=o,mark options={solid,fill=white}}
\addlegendentry{$\tau=2$, $H(\boldsymbol\pi_2)=3.084$ bits}
\addlegendimage{black,only marks,mark=o,mark options={solid,fill=white}}
\addlegendentry{$\tau=3$, $H(\boldsymbol\pi_3)=2.721$ bits}
\addlegendimage{black,solid,mark=none}
\addlegendentry{KKT prediction}
\addlegendimage{black,only marks,mark=x}
\addlegendentry{Probability floor}
\end{semilogyaxis}
\end{tikzpicture}
\caption{Optimized pattern probabilities ($\pi_{\tau,r}$) versus marginal pattern costs ($\psi_{\tau,r}$) for fixed nominal PS-SLP precoding vectors.}
\label{fig:satellite_gibbs}
\end{figure}

Fig.~\ref{fig:satellite_gibbs} reveals how this benefit arises by showing the optimized pattern probabilities of the nominal PS-SLP design versus their costs $\psi_{\tau,r}$ in \eqref{eq:psi}. In every block, the probabilities above the floor lie on a straight line in logarithmic scale and match the prediction of Remark~\ref{rem:gibbs}. Patterns that are inexpensive to deliver with low leakage are used most often, with probabilities up to about $0.4$, whereas the most expensive ones are driven to the floor $\pi_{\min}$, which is reached by two patterns in block~2 and four in block~3. For comparison, the uniform design assigns probability $1/16\approx0.06$ to every pattern. The three lines are parallel, i.e., all blocks share the same temperature $\vartheta$, as expected from Remark~\ref{rem:gibbs} when only the session entropy constraint $\mathrm C_5$ is active. The blocks therefore differ only in their costs, which are determined by the geometry. The patterns of block~3 are more expensive overall, so this block is assigned a more concentrated distribution, with entropy $2.72$~bits, compared to $3.20$ and $3.08$~bits for blocks~1 and~2. These values average exactly to $\bar H=3$~bits, so the session constraint allocates entropy to the blocks where its impact on the covertness certificate is smallest. This also explains the trend in Fig.~\ref{fig:satellite_2}, since relaxing $\bar H$ lowers $\vartheta$, steepens the lines, and makes the expensive patterns even rarer.

\begin{table}[t]
\centering
\caption{PS-SLP performance under additional constraints.}
\label{tab:additional_constraints}
\setlength{\tabcolsep}{3pt}
\renewcommand{\arraystretch}{1.15}
\scriptsize
\begin{tabular}{lccccc}
\hline
Case
& $\epsilon_{\min}^{\mathrm{cert}}$
& $E_{\mathrm{used}}/E_{\mathrm{sess}}$
& $H_1$ & $H_2$ & $H_3$ \\
\hline
Nominal
& 0.03422 & 0.822 & 3.195 & 3.084 & 2.721 \\
Reduced energy
& 0.03623 & 1.000 & 3.278 & 3.025 & 2.697 \\
Per-block entropy
& 0.03498 & 0.826 & 3.000 & 3.000 & 3.000 \\
Channel-model error
& 0.03855 & 0.780 & 3.252 & 3.007 & 2.742 \\
\hline
\end{tabular}
\end{table}

Table~\ref{tab:additional_constraints} quantifies the effect of additional constraints on the nominal PS-SLP design, where $H_\tau=H(\boldsymbol\pi_\tau)$ is the entropy of block $\tau$ in bits and $E_{\mathrm{used}}$ is the expected transmit energy. Each case changes one setting of the nominal design, which uses $\beta_E=1$, $H_{\mathrm{floor}}=0$, and $\delta_{\tau,q}^{\mathrm{mod}}=0$, and reoptimizes it. Setting $H_{\mathrm{floor}}=\bar H=3$~bits yields the per-block entropy benchmark of Section~\ref{subsec:benchmarks}. It forces block~3, whose geometry is the least favorable, to $3$~bits and increases $\epsilon_{\min}^{\mathrm{cert}}$ by about $2\%$, which quantifies the benefit of allocating entropy across blocks. Halving the energy budget ($\beta_E=0.5$) makes the energy constraint active and increases $\epsilon_{\min}^{\mathrm{cert}}$ by about $6\%$. Finally, a channel-model error with $\delta_{\tau,q}^{\mathrm{mod}}=10^{-3}\beta_{\tau,q}^{\max}$, i.e., $0.1\%$ of Willie's channel amplitude, increases $\epsilon_{\min}^{\mathrm{cert}}$ by about $13\%$. Since the corresponding margin grows with $\|\mathbf v_{\tau,r}\|$, the design also lowers its energy use, from $82\%$ to $78\%$ of the budget. All cases remain feasible and the certified level changes only moderately.

\begin{figure}[t]
\centering
\begin{tikzpicture}
\begin{axis}[
    width=0.77\linewidth,
    height=0.69\linewidth,
    ylabel={$\epsilon_{\min}^{\mathrm{cert}}$},
    ymin=0,
    scaled x ticks=false,
    scaled y ticks=false,
    x tick label style={
        /pgf/number format/fixed,
        /pgf/number format/precision=0,
        font=\scriptsize
    },
    y tick label style={
        /pgf/number format/fixed,
        /pgf/number format/precision=1,
        font=\scriptsize
    },
    grid=both,
    minor grid style={gray!20},
    major grid style={gray!45},
    label style={font=\scriptsize},
    unbounded coords=jump,
    legend columns=2,
    legend cell align=left,
    legend style={
        font=\scriptsize,
        at={(1,1)},
        anchor=north east,
        draw=black,
        fill=white,
        fill opacity=0.8,
        text opacity=1,
        inner sep=2pt,
        row sep=-1pt,
        column sep=3pt
    },
    legend image post style={line width=0.7pt,scale=0.8},
 xlabel={$\theta_{\mathrm{el}}^{\mathrm{mid}}$ (deg)},
 xmin=60,xmax=80,xtick={60,65,70,75,80},ymax=0.5,ytick={0,0.2,0.4,0.6,0.8,1}
]
\addplot[gray!65,densely dotted,line width=0.8pt,forget plot] coordinates {(60,1) (80,1)};

\addplot[
    blue, mark=o,
    mark options={solid,fill=white},
    mark size=1.7pt,
    line width=1pt,
    solid,
    forget plot
]
table[x=elev_PS_SLP_g2,y=PS_SLP_g2,col sep=space]
{Figures4/fig3_data.dat};

\addplot[
    blue, mark=square,
    mark options={solid,fill=white},
    mark size=1.7pt,
    line width=1pt,
    dashed,
    forget plot
]
table[x=elev_U_SLP_g2,y=U_SLP_g2,col sep=space]
{Figures4/fig3_data.dat};

\addplot[
    red, mark=o,
    mark options={solid,fill=white},
    mark size=1.7pt,
    line width=1pt,
    solid,
    forget plot
]
table[x=elev_PS_SLP_g5,y=PS_SLP_g5,col sep=space]
{Figures4/fig3_data.dat};

\addplot[
    red, mark=square,
    mark options={solid,fill=white},
    mark size=1.7pt,
    line width=1pt,
    dashed,
    forget plot
]
table[x=elev_U_SLP_g5,y=U_SLP_g5,col sep=space]
{Figures4/fig3_data.dat};

\addplot[
    black, mark=o,
    mark options={solid,fill=white},
    mark size=1.7pt,
    line width=1pt,
    solid,
    forget plot
]
table[x=elev_PS_SLP_g8,y=PS_SLP_g8_preview,col sep=space]
{Figures4/fig3_data.dat};

\addplot[
    black, mark=square,
    mark options={solid,fill=white},
    mark size=1.7pt,
    line width=1pt,
    dashed,
    mark indices={1,2,3,4,5,6,7,8,9},
    forget plot
]
table[x=elev_U_SLP_g8,y=U_SLP_g8,col sep=space]
{Figures4/fig3_data.dat};


\addlegendimage{black,solid,mark=o,mark options={solid,fill=white}}
\addlegendentry{PS-SLP}
\addlegendimage{black,dashed,mark=square,mark options={solid,fill=white}}
\addlegendentry{U-SLP}
\addlegendimage{blue,solid,mark=none,line width=0.9pt}
\addlegendentry{$\gamma_{\mathrm{req}}=2$ dB}
\addlegendimage{red,solid,mark=none,line width=0.9pt}
\addlegendentry{$\gamma_{\mathrm{req}}=5$ dB}
\addlegendimage{black,solid,mark=none,line width=0.9pt}
\addlegendentry{$\gamma_{\mathrm{req}}=8$ dB}
\end{axis}
\end{tikzpicture}
\caption{Minimum certified covertness level ($\epsilon_{\min}^{\mathrm{cert}}$) versus the mid-session satellite elevation for different CI QoS thresholds.}
\label{fig:satellite_3}
\end{figure}
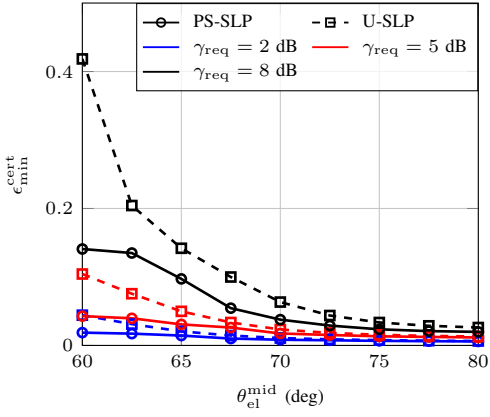


Fig.~\ref{fig:satellite_3} shows the effect of the satellite elevation at the session midpoint, from $60^\circ$ to $80^\circ$, for $\gamma_{\mathrm{req}}\in\{2,5,8\}$~dB, with the precoding vectors redesigned at each elevation. A higher elevation shortens slant ranges, which reduces the power needed to serve the Bobs but also strengthens Willie's channel. In this deployment, it also increases the separation between the Bobs and Willie's region in direction cosine, improving leakage suppression. This improved spatial discrimination and the reduced transmit-power requirement outweigh Willie's stronger channel, so the certified level decreases with elevation for all thresholds.
The decrease is strongest at $\gamma_{\mathrm{req}}=8$~dB, where the certified level drops from $0.14$ to $0.020$ for PS-SLP and from $0.42$ to $0.026$ for U-SLP between $60^\circ$ and $80^\circ$. PS-SLP outperforms U-SLP at every elevation, with the largest gain at low elevation, where spatial discrimination is weakest. At $60^\circ$, it reduces the certified level by $58$--$66\%$ across the three thresholds. At $80^\circ$, the reduction is about $10\%$ for $\gamma_{\mathrm{req}}\in\{2,5\}$~dB but $24\%$ for $\gamma_{\mathrm{req}}=8$~dB. Shaping is therefore most valuable when the geometry alone does not separate the Bobs from Willie and when the QoS requirement is strict.

\begin{figure}[t]
\centering
\begin{tikzpicture}
\begin{axis}[
    width=0.77\linewidth,
    height=0.69\linewidth,
    ylabel={$\epsilon_{\min}^{\mathrm{cert}}$},
    ymin=0,
    scaled x ticks=false,
    scaled y ticks=false,
    x tick label style={
        /pgf/number format/fixed,
        /pgf/number format/precision=0,
        font=\scriptsize
    },
    y tick label style={
        /pgf/number format/fixed,
        /pgf/number format/precision=1,
        font=\scriptsize
    },
    grid=both,
    minor grid style={gray!20},
    major grid style={gray!45},
    label style={font=\scriptsize},
    unbounded coords=jump,
    legend columns=2,
    legend cell align=left,
    legend style={
        font=\scriptsize,
        at={(0.00,1)},
        anchor=north west,
        draw=black,
        fill=white,
        fill opacity=0.8,
        text opacity=1,
        inner sep=2pt,
        row sep=-1pt,
        column sep=3pt
    },
    legend image post style={line width=0.7pt,scale=0.8},
    xlabel={Number of active channel uses, $n_{\mathrm{obs}}$},
    xmin=0,xmax=90,
    xtick={0,15,30,45,60,75,90},
    ymax=0.6,
    ytick={0,0.2,0.4,0.6,0.8,1}
]
\addplot[gray!65,densely dotted,line width=0.8pt,forget plot] coordinates {(0,1) (90,1)};

\addplot[
    blue, mark=o,
    mark options={solid,fill=white},
    mark size=1.7pt,
    line width=1pt,
    solid,
    forget plot
]
table[x=n_PS_SLP_p015,y=PS_SLP_p015,col sep=space]
{Figures4/fig5_data.dat};

\addplot[
    blue, mark=square,
    mark options={solid,fill=white},
    mark size=1.7pt,
    line width=1pt,
    dashed,
    forget plot
]
table[x=n_U_SLP_p015,y=U_SLP_p015,col sep=space]
{Figures4/fig5_data.dat};

\addplot[
    red, mark=o,
    mark options={solid,fill=white},
    mark size=1.7pt,
    line width=1pt,
    solid,
    forget plot
]
table[x=n_PS_SLP_p030,y=PS_SLP_p030,col sep=space]
{Figures4/fig5_data.dat};

\addplot[
    red, mark=square,
    mark options={solid,fill=white},
    mark size=1.7pt,
    line width=1pt,
    dashed,
    forget plot
]
table[x=n_U_SLP_p030,y=U_SLP_p030,col sep=space]
{Figures4/fig5_data.dat};

\addplot[
    black, mark=o,
    mark options={solid,fill=white},
    mark size=1.7pt,
    line width=1pt,
    solid,
    forget plot
]
table[x=n_PS_SLP_p080,y=PS_SLP_p080,col sep=space]
{Figures4/fig5_data.dat};

\addplot[
    black, mark=square,
    mark options={solid,fill=white},
    mark size=1.7pt,
    line width=1pt,
    dashed,
    forget plot
]
table[x=n_U_SLP_p080,y=U_SLP_p080,col sep=space]
{Figures4/fig5_data.dat};

\addlegendimage{black,solid,mark=o,mark options={solid,fill=white}}
\addlegendentry{PS-SLP}
\addlegendimage{black,dashed,mark=square,mark options={solid,fill=white}}
\addlegendentry{U-SLP}
\addlegendimage{blue,solid,mark=none,line width=0.9pt}
\addlegendentry{$P^{\mathrm{pk}}=0.15$ W}
\addlegendimage{red,solid,mark=none,line width=0.9pt}
\addlegendentry{$P^{\mathrm{pk}}=0.30$ W}
\addlegendimage{black,solid,mark=none,line width=0.9pt}
\addlegendentry{$P^{\mathrm{pk}}=0.80$ W}
\end{axis}
\end{tikzpicture}
\caption{Minimum certified covertness level ($\epsilon_{\min}^{\mathrm{cert}}$) versus the number of active channel uses for different per-RF-chain peak-power limits, with $P^{\mathrm{av}}=P^{\mathrm{pk}}/2$.}
\label{fig:satellite_5}
\end{figure}
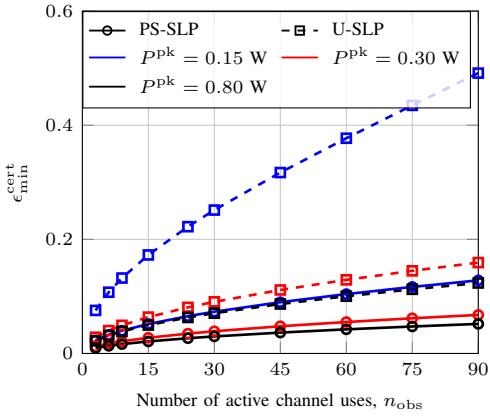

Fig.~\ref{fig:satellite_5} shows the certified level versus the number of active channel uses for $P^{\mathrm{pk}}\in\{0.15,0.30,0.80\}$~W, with $P^{\mathrm{av}}=P^{\mathrm{pk}}/2$. Here, $T=3$ and the block centers are fixed, while $L$ varies from $1$ to $30$, i.e., $n_{\mathrm{obs}}=3L$ between $3$ and $90$, and the energy budget scales with $n_{\mathrm{obs}}$. The certified level increases as Willie accumulates observations and, as predicted by Remark~\ref{rem:length_scaling}, follows exactly the curve $\frac12\sqrt{e^{c\,n_{\mathrm{obs}}}-1}$, since the same design is optimal for every $L$. For the nominal design ($P^{\mathrm{pk}}=0.5$~W), $c\approx1.6\times10^{-4}$, so by \eqref{eq:nmax} sessions of up to about $250$ channel uses can be certified for $\epsilon=0.1$, and up to about $1970$ for $\epsilon=0.3$. A larger power budget lowers the certified level, since it leaves more freedom to satisfy the CI requirements while suppressing leakage. As the budgets are upper limits, this does not require more transmit power, and the benefit diminishes beyond $0.30$~W. The gain of shaping is largest when power is scarce. At $n_{\mathrm{obs}}=30$, PS-SLP reduces the certified level relative to U-SLP by $71\%$ for $P^{\mathrm{pk}}=0.15$~W, and by about $57\%$ for $0.30$ and $0.80$~W.

\section{Conclusion}
\label{Conc}

This paper investigated covert multiuser LEO satellite downlinks, where a passive warden at an unknown but fixed ground location accumulates observations as the satellite moves. We proposed the PS-SLP and PS-BLP designs, which jointly optimize the precoding and the probabilities of relative multiuser patterns, while every legitimate user still observes uniformly distributed PSK symbols. We derived the exact chi-square divergence of the warden's observation over the complete session and, by bounding the leakage over a continuous location uncertainty region, obtained covertness certificates that hold at every warden position. The certified covertness level was minimized under robust CI, power, and entropy constraints through an alternating optimization whose updates are equivalent to convex conic problems. The optimized probabilities were shown to follow a Maxwell--Boltzmann-type law, which allocates the session entropy across blocks according to the satellite geometry. Numerical results showed that shaping substantially reduces the certified covertness level of both precoding structures, and that its benefit is largest when precoding alone has limited freedom to suppress the leakage, i.e., under large location uncertainty, at low elevation, under strict QoS requirements, and under scarce power. These findings establish the symbol statistics as a design resource for covertness in their own right. A covert LEO transmitter should therefore control not only where its signal is radiated but also how often each multiuser symbol combination is sent.

\appendices
\section{Proof of Proposition~\ref{prop:exact}}
\label{app:exact}
Expanding the square in \eqref{eq:chi} and using
$\int_{\mathbb C}p_{1,\tau}(y)\,dy
=\int_{\mathbb C}p_0(y)\,dy=1$ gives
$D_\tau=\int_{\mathbb C}p_{1,\tau}^2(y)/p_0(y)\,dy$.
For the Gaussian densities in \eqref{eq:gaussian_density}, completing the square yields
\begin{equation}
\int_{\mathbb C}\frac{p_a(y)p_b(y)}{p_0(y)}\,dy
=\exp\!\left(\frac{2\Re\{a^*b\}}{\sigma_w^2}\right).
\label{eq:gaussian_overlap}
\end{equation}
Substituting \eqref{eq:mixture}, we obtain
\begin{equation}
\begin{aligned}
D_\tau(\mathbf p_W)
={}&\sum_{r=1}^{R}\sum_{t=1}^{R}
\frac{\pi_{\tau,r}\pi_{\tau,t}}{\mathcal M^2}\\
&\times
\sum_{\zeta_1,\zeta_2\in\mathcal S}
\exp\!\left(
\frac{2}{\sigma_w^2}
\Re\{\zeta_1^*\zeta_2\mu_{\tau,r}^*\mu_{\tau,t}\}
\right).
\end{aligned}
\end{equation}
For each $\upsilon\in\mathcal S$, exactly $\mathcal M$ of the $\mathcal M^2$ pairs $(\zeta_1,\zeta_2)$ satisfy $\zeta_1^*\zeta_2=\upsilon$. Thus, with $z=\mu_{\tau,r}^*\mu_{\tau,t}/\sigma_w^2$,
\begin{equation}
\frac{1}{\mathcal M^2}
\sum_{\zeta_1,\zeta_2\in\mathcal S}
e^{2\Re\{\zeta_1^*\zeta_2z\}}
=\frac{1}{\mathcal M}
\sum_{\upsilon\in\mathcal S}e^{2\Re\{\upsilon z\}}
=\Phi_{\mathcal M}(z).
\end{equation}
Substitution proves \eqref{eq:exactD}.

\section{Proof of Lemma~\ref{lem:product}}
\label{app:kernel}
Since $\mathcal M$ is divisible by four, the PSK rotations can be grouped into $D_{\mathrm{rot}}=\mathcal M/4$ QPSK quadruples, giving
\begin{equation}
\Phi_{\mathcal M}(z)
=\frac{1}{D_{\mathrm{rot}}}\sum_{d=0}^{D_{\mathrm{rot}}-1}
\Phi_4(e^{\mathrm j2\pi d/\mathcal M}z),
\label{eq:kernel_grouping}
\end{equation}
where the auxiliary QPSK kernel is
\begin{equation}
\Phi_4(w)=\frac12
\left[\cosh(2\Re\{w\})+\cosh(2\Im\{w\})\right].
\label{eq:aux_qpsk_kernel}
\end{equation}
For $w=x+\mathrm jy$, the identity
$\cosh(2x)=1+2\sinh^2x$ yields
$
\Phi_4(w)-1=\sinh^2x+\sinh^2y\geq0.
$
The power series of $\sinh^2(\cdot)$ has nonnegative even-order coefficients. Since $x^{2\iota}+y^{2\iota}\leq(x^2+y^2)^\iota$ for every integer $\iota\geq1$, applying this inequality term-wise gives
$
\Phi_4(w)-1
\leq\sinh^2\!\left(\sqrt{x^2+y^2}\right)
=\sinh^2(|w|).
$
All rotated arguments in \eqref{eq:kernel_grouping} have modulus $|z|$. Averaging and using
$|z|^2\leq t_1t_2$ therefore yield
\begin{equation}
0\leq\Phi_{\mathcal M}(z)-1
\leq\sinh^2(|z|)
\leq\sinh^2(\sqrt{t_1t_2}).
\label{eq:psk_kernel_radial}
\end{equation}
Using the power series of $\sinh$ and the Cauchy--Schwarz inequality,
\begin{equation}
\begin{gathered}
\sinh(\sqrt{t_1t_2})
=\sum_{\iota=0}^{\infty}
\frac{t_1^{\iota+1/2}t_2^{\iota+1/2}}{(2\iota+1)!}\\
\leq
\left[
\left(\sum_{\iota=0}^{\infty}
\frac{t_1^{2\iota+1}}{(2\iota+1)!}\right)
\left(\sum_{\iota=0}^{\infty}
\frac{t_2^{2\iota+1}}{(2\iota+1)!}\right)
\right]^{1/2}\\
=\sqrt{\sinh(t_1)\sinh(t_2)}.
\end{gathered}
\label{eq:sinh_product_cs}
\end{equation}
Squaring \eqref{eq:sinh_product_cs} and combining it with \eqref{eq:psk_kernel_radial} proves the hyperbolic-sine bound. For $\mathcal M=4$, it is attained when $t_1=t_2$ and $z=t_1$. Finally,
$e^t-1-\sinh(t)=\cosh(t)-1\geq0$ for $t\geq0$,
which proves the exponential bound.

\section{Proof of Theorem~\ref{thm:exp}}
\label{app:exp}
We apply Lemma~\ref{lem:product} to every term of \eqref{eq:exactD}, accumulate the blocks through \eqref{eq:session_log}, and take the worst cell. For fixed $\tau$ and $\mathbf p_W\in\mathcal W_q$, let $z_{r,t}=\mu_{\tau,r}^{*}(\mathbf p_W)\mu_{\tau,t}(\mathbf p_W)/\sigma_w^2$. From \eqref{eq:rho_guarantee}, $|z_{r,t}|^2\leq\rho_{\tau,r,q}\rho_{\tau,t,q}$. Since $\sum_{r,t}\pi_{\tau,r}\pi_{\tau,t}=1$, applying Lemma~\ref{lem:product} with $t_1=\rho_{\tau,r,q}$ and $t_2=\rho_{\tau,t,q}$ to \eqref{eq:exactD} gives
\begin{align}
\chi_\tau^2(\mathbf p_W)
&=\sum_{r,t}\pi_{\tau,r}\pi_{\tau,t}\bigl[\Phi_{\mathcal M}(z_{r,t})-1\bigr]\nonumber\\
&\leq\sum_{r,t}\pi_{\tau,r}\pi_{\tau,t}\sinh(\rho_{\tau,r,q})\sinh(\rho_{\tau,t,q})\nonumber\\
&=\left(M_{\tau,q}^{\mathrm{sinh}}\right)^2\leq M_{\tau,q}^2,
\label{eq:exp_proof_chain}
\end{align}
where the last inequality is the last one in \eqref{eq:product_bound}. This proves \eqref{eq:exo_one}.

Moreover, every $\mathbf p_W\in\mathcal A_W$ lies in some cell $\mathcal W_q$, since the cells cover $\mathcal A_W$. For this cell, \eqref{eq:session_log} and \eqref{eq:exp_proof_chain} give
\begin{equation}
\ln\!\left(1+\chi_{\rm sess}^2(\mathbf p_W)\right)
\leq L\sum_{\tau=1}^{T}\ln\!\left[1+\left(M_{\tau,q}^{\mathrm{sinh}}\right)^2\right]
\leq C_{\mathrm{sinh}}.
\end{equation}
Finally, $\ln(1+x)\leq x$ for $x\geq0$ and $M_{\tau,q}^{\mathrm{sinh}}\leq M_{\tau,q}$ give $C_{\mathrm{sinh}}\leq C_{\mathrm E}$, which proves \eqref{eq:CE_chain}.

\section{Proof of Proposition~\ref{prop:separable}}
\label{app:separable}
We first establish the per-block bound, then its session counterpart, and finally the ordering $C_{\mathrm{sinh}}\leq C_{\mathrm S}$.

For $w=x+\mathrm jy$, the auxiliary kernel $\Phi_4$
defined in \eqref{eq:aux_qpsk_kernel} satisfies
\begin{equation}
\begin{aligned}
\Phi_4(w)
&=1+\sum_{\iota=1}^{\infty}
\frac{2^{2\iota-1}}{(2\iota)!}
(x^{2\iota}+y^{2\iota})\\
&\leq1+\sum_{\iota=1}^{\infty}
\frac{|w|^{2\iota}}{\iota!}
=e^{|w|^2},
\end{aligned}
\label{eq:aux_kernel_exp}
\end{equation}
where the inequality follows from
$2^{2\iota-1}\iota!\leq(2\iota)!$ and
$x^{2\iota}+y^{2\iota}\leq(x^2+y^2)^\iota$
for every integer $\iota\geq1$.
Applying \eqref{eq:aux_kernel_exp} to each term of \eqref{eq:kernel_grouping}, with $D_{\mathrm{rot}}=\mathcal M/4$, yields
\begin{equation}
\Phi_{\mathcal M}(z)
\leq\frac{1}{D_{\mathrm{rot}}}\sum_{d=0}^{D_{\mathrm{rot}}-1}
e^{|e^{\mathrm j2\pi d/\mathcal M}z|^2}
=e^{|z|^2}.
\label{eq:psk_kernel_exp}
\end{equation}
For $\mathbf p_W\in\mathcal W_q$,
\eqref{eq:rho_guarantee} gives
$|\mu_{\tau,r}^*\mu_{\tau,t}|^2/\sigma_w^4
\leq\rho_{\tau,r,q}\rho_{\tau,t,q}$.
Substituting into \eqref{eq:exactD},
\begin{equation}
\begin{aligned}
D_\tau(\mathbf p_W)
&\leq\sum_{r,t}\pi_{\tau,r}\pi_{\tau,t}
e^{\rho_{\tau,r,q}\rho_{\tau,t,q}}\\
&\leq\sum_{r,t}\pi_{\tau,r}\pi_{\tau,t}
e^{(\rho_{\tau,r,q}^2+\rho_{\tau,t,q}^2)/2}\\
&=\left[
\sum_{r=1}^{R}\pi_{\tau,r}
e^{\rho_{\tau,r,q}^2/2}
\right]^2,
\end{aligned}
\label{eq:psk_separable_factorization}
\end{equation}
where the second inequality uses $ab\leq(a^2+b^2)/2$. Taking logarithms proves \eqref{eq:dbar}. Accumulating the blocks and maximizing over the cells gives the corresponding session bound.  

To show $C_{\mathrm{sinh}}\leq C_{\mathrm S}$, let $\mathbb E[\cdot]$ denote the average over $\boldsymbol\pi_\tau$ and write $\rho=\rho_{\tau,\mathsf r,q}$. By the Cauchy--Schwarz inequality,
$(\mathbb E[\cosh\rho])^2 - (\mathbb E[\sinh\rho])^2=\mathbb E[e^\rho]\mathbb E[e^{-\rho}]\geq1$.
Thus,
$1+(M_{\tau,q}^{\mathrm{sinh}})^2
\leq(\mathbb E[\cosh\rho])^2$.
Since $\cosh\rho\leq e^{\rho^2/2}$,
\begin{equation}
\begin{aligned}
\ln\!\left[1+(M_{\tau,q}^{\mathrm{sinh}})^2\right]\leq2\ln\mathbb E[\cosh\rho]\leq2\ln\mathbb E[e^{\rho^2/2}]=\bar d_{\tau,q}.
\end{aligned}
\label{eq:sinh_separable_ordering}
\end{equation}
Summing over the blocks and maximizing over the cells gives $C_{\mathrm{sinh}}\leq C_{\mathrm S}$. Together with Theorem~\ref{thm:exp}, this proves \eqref{eq:separable_session}.

\section{Proof of Proposition~\ref{prop:precoding_convex}}
\label{app:precoding_cones}
First, assume that \eqref{eq:C2_conic} holds. Since \eqref{eq:C2_conic_a} is imposed at every reference direction, it implies $A_{\tau,r,q}\leq a_{\tau,r,q}$, including the channel-model margin when present. The rotated cone \eqref{eq:C2_conic_b} is equivalent to $a_{\tau,r,q}^2\leq\sigma_{w,\min}^2b_{\tau,r,q}$, so that $\rho_{\tau,r,q}=A_{\tau,r,q}^2/\sigma_{w,\min}^2\leq b_{\tau,r,q}$. The exponential cone \eqref{eq:C2_conic_c} is equivalent to $e^{b_{\tau,r,q}}-1\leq w_{\tau,r,q}$, and since $e^x$ is increasing, $e^{\rho_{\tau,r,q}}-1\leq w_{\tau,r,q}$. As the probabilities are nonnegative, \eqref{eq:C2_conic_d} then gives
\begin{equation}
M_{\tau,q}=\sum_{r=1}^{R}\pi_{\tau,r}^{(i)}\bigl(e^{\rho_{\tau,r,q}}-1\bigr)\leq\sum_{r=1}^{R}\pi_{\tau,r}^{(i)}w_{\tau,r,q}\leq t_{\tau,q}.
\end{equation}
Finally, \eqref{eq:C2_conic_e} is equivalent to $L\|\mathbf t_q\|^2\leq C$. Since $0\leq M_{\tau,q}\leq t_{\tau,q}$, this yields $L\sum_{\tau}M_{\tau,q}^2\leq C$ for every $q$, i.e., $\mathrm C_2$.

Conversely, assume that $\mathrm C_2$ holds and choose $a_{\tau,r,q}=A_{\tau,r,q}$,
$b_{\tau,r,q}=\rho_{\tau,r,q}$,
$w_{\tau,r,q}=e^{\rho_{\tau,r,q}}-1$, and
$t_{\tau,q}=M_{\tau,q}$.
These choices are nonnegative and satisfy \eqref{eq:C2_conic_a}--\eqref{eq:C2_conic_d}, with equality in \eqref{eq:C2_conic_a} for the maximizing reference direction, while \eqref{eq:C2_conic_e} reduces to $\mathrm C_2$. This completes the proof.

\section{Proof of Proposition~\ref{prop:probability_convex}}
\label{app:probability_cones}
Since $\pi_{\tau,r}\geq\pi_{\min}>0$, the definition \eqref{eq:expcone} shows that \eqref{eq:entropy_conic_a} is equivalent to $\pi_{\tau,r}e^{-h_{\tau,r}/\pi_{\tau,r}}\leq1$, i.e., to $h_{\tau,r}\geq\pi_{\tau,r}\ln\pi_{\tau,r}$. First, assume that \eqref{eq:entropy_conic} holds. Combining this inequality with \eqref{eq:entropy_conic_b} and \eqref{eq:entropy_conic_c} gives
\begin{equation}
\begin{aligned}
L\sum_{\tau=1}^{T}\sum_{r=1}^{R}
\pi_{\tau,r}\ln\pi_{\tau,r}
&\leq-n_{\mathrm{obs}}\bar H\ln2,\\
\sum_{r=1}^{R}\pi_{\tau,r}\ln\pi_{\tau,r}
&\leq-H_{\mathrm{floor}}\ln2,\quad\forall\tau.
\end{aligned}
\label{eq:entropy_bounds_appendix}
\end{equation}
Since $H(\boldsymbol\pi_\tau)=-\sum_r\pi_{\tau,r}\ln\pi_{\tau,r}/\ln2$, these are exactly $\mathrm C_5$ and $\mathrm C_6$.
Conversely, if $\mathrm C_5$ and $\mathrm C_6$ hold, choosing $h_{\tau,r}=\pi_{\tau,r}\ln\pi_{\tau,r}$ satisfies \eqref{eq:entropy_conic_a} with equality, and \eqref{eq:entropy_conic_b} and \eqref{eq:entropy_conic_c} reduce to $\mathrm C_5$ and $\mathrm C_6$. This completes the proof.

\bibliographystyle{IEEEtran}
\bibliography{references}
\end{document}